\documentclass[12pt, reqno]{amsart}

\usepackage{graphics, stackrel}
\usepackage{amsmath, amssymb, amsthm}
\usepackage{graphicx}
\usepackage{verbatim}
\usepackage{amsfonts}

\usepackage{natbib}

\usepackage{enumitem}

\usepackage[mathscr]{euscript}
\usepackage{stmaryrd} 

\usepackage[xcharter]{newtxmath}

\usepackage{caption}
\usepackage{subcaption}

\usepackage{booktabs}
\usepackage{makecell}
\usepackage{threeparttable}
\usepackage{siunitx}
\usepackage{fancyvrb}
\usepackage[usenames, dvipsnames, svgnames,table]{xcolor}
\usepackage{mdwlist}

\usepackage{enumitem}
\setlist[enumerate]{itemsep=2pt,topsep=3pt}
\setlist[itemize]{itemsep=2pt,topsep=3pt}
\setlist[enumerate,1]{label=(\alph*)}

\usepackage{mathrsfs}
\usepackage{bbm}

\usepackage[letterpaper, left=1.25in, right=1.25in, top=1.25in, bottom=1.25in, includehead, includefoot]{geometry}  
\usepackage[citecolor=blue, colorlinks=true, urlcolor=blue, linkcolor=blue]{hyperref}

\renewcommand{\phi}{\varphi}
\renewcommand{\epsilon}{\varepsilon}

\renewcommand{\leq}{\leqslant}
\renewcommand{\geq}{\geqslant}

\usepackage[ruled, linesnumbered]{algorithm2e}

\newcommand\blfootnote[1]{%
  \begingroup
  \renewcommand\thefootnote{}\footnote{#1}%
  \addtocounter{footnote}{-1}%
  \endgroup
}

\newcommand{\setntn}[2]{ \{ #1 : #2 \} }
\newcommand{\fore}{\therefore \quad}

\newcommand{\1}{\mathbbm 1}

\newcommand*\diff{\mathop{}\!\mathrm{d}}

\newcommand{\cf}{\mathscr C}

\newcommand{\sS}{\mathcal S}
\newcommand{\bB}{\mathscr B}

\newcommand{\mM}{\mathscr M}
\newcommand{\eE}{\mathcal E}

\newcommand{\pP}{\mathscr P}

\newcommand{\RR}{\mathbbm R}

\newcommand{\NN}{\mathbbm N}

\newcommand{\PP}{\mathbbm P}

\newcommand{\EE}{\mathbbm E}

\newif\ifdraftmarks
\draftmarkstrue

\theoremstyle{plain}
\newtheorem{theorem}{Theorem}[section]
\newtheorem{corollary}[theorem]{Corollary}
\newtheorem{lemma}[theorem]{Lemma}
\newtheorem{proposition}[theorem]{Proposition}

\theoremstyle{definition}
\newtheorem{definition}{Definition}[section]

\newtheorem{example}{Example}[section]
\newtheorem{remark}{Remark}[section]

\newtheorem{assumption}{Assumption}[section]

\usepackage[normalem]{ulem}

\begin{document}


\title{}

\begin{center}
  \Large Firm Entry and Exit with Unbounded Productivity\blfootnote{The
  authors thank Hugo Hopenhayn and Takashi Kamihigashi for valuable comments
  on an earlier draft of this manuscript (which was titled \emph{Power laws
  without Gibrat's law}). Financial support from ARC grant FT160100423 and
  Alfred P. Sloan Foundation grant G-2016-7052 is gratefully acknowledged.}

    \vspace{1em}

  \large Akshay Shanker\textsuperscript{a} and John
  Stachurski\textsuperscript{b} \par \bigskip

  \small \textsuperscript{a}Econ-ARK \& School of Economics, University of New
  South Wales, Sydney \bigskip \vspace{-1em}

  \textsuperscript{b} National Graduate Institute for Policy Studies, Tokyo

  \normalsize \vspace{1em} \today
\end{center}

\begin{abstract}
    In Hopenhayn's (1992) entry-exit model productivity is bounded, implying
    that the predicted firm size distribution cannot match the power law tail
    observable in the data. In this paper we remove the boundedness assumption
    and, in this more general setting, provide an exact characterization of
    existence of stationary equilibria, as well as a novel sufficient
    condition for existence based on treating production as a Lyapunov
    function. We also provide new representations of the rate of entry and
    aggregate supply.  Market clearing rules out stationary competitive
    equilibria at or below Zipf's law.
    Under monopolistic competition where the demand for each firm's 
    variety bounds its sales, a stationary equilibrium exists for every tail 
    index of productivity. In an example with additive preferences, unbounded
    markups inherit the Pareto tail of productivity.

\end{abstract}

\maketitle

\section{Introduction}

Observed productivity growth is driven by the decisions of firms. Existing
firms innovate, while new firms disrupt the status quo by bringing fresh
ideas. The specifics of this process have far-reaching implications for long
run growth in per capita output, output volatility, employment, the income
distribution, the wealth distribution, and the development of technology and
human capital.

One of the most important models of firm dynamics developed in the last few
decades is the entry-exit model of
\cite{hopenhayn1992entry}.\footnote{\cite{hopenhayn1992entry} forms a
cornerstone of modern quantitative economics and researchers have extended the
core ideas to study a broad range of topics, from technical change and
development to aggregate volatility and business cycle fluctuations. The
literature is very large. Well known examples include \cite{hopenhayn1993job},
\cite{clementi2006theory}, \cite{ericson1995markov},
\cite{luttmer2011mechanics}, \cite{acemoglu2015innovation},
\cite{cao2022firm}, and \cite{carvalho2019large}. Recent quantitative examples
include \cite{sterk2021nature}, \cite{berlingieri2025empirical} and
\cite{hubmer2026not}.} In building his model, \cite{hopenhayn1992entry} makes
one key technical assumption to streamline his analysis: firm productivity is
bounded. Much of the subsequent theoretical and quantitative work follows his
assumption. At the same time, under the bounded productivity assumption the
model of \cite{hopenhayn1992entry} cannot match a central feature of the data:
the firm size distribution is extremely heavy tailed -- in fact a power law,
with a Pareto tail coefficient near unity.\footnote{In other words, for some
measure of firm size $S$, there are positive constants $k$ and $\alpha$ close
to one such that $\PP\{S > s\} \approx k s^{-\alpha}$ for large $s$. Zipf's
law is the case $\alpha = 1$. A well known reference is \cite{axtell2001zipf}.
The power law finding has been replicated in many studies. See, for example,
\cite{gaffeo2003size}, who treats the G7 economies, as well as
\cite{cirillo2009upper}, \cite{kang2011changes} and \cite{zhang2009zipf}, who
use Italian, Korean and Chinese data respectively, and \cite{chen2026growth},
who uses tabulated size data across countries and years.}
This power law property is significant for a range of aggregate outcomes. For
example, idiosyncratic firm-level shocks generate substantial aggregate
volatility when the firm size distribution has a power law.\footnote{See, for
example, \cite{nirei2006threshold}, \cite{gabaix2011granular}, and
\cite{carvalho2019large}.} Moreover, the right tail of the firm size
distribution affects the income and wealth distributions, partly due to the
high concentration of firm ownership and entrepreneurial equity.\footnote{See,
for example, \cite{benhabib2018skewed}. The impact of capital income on income
and wealth dispersion has risen in recent years, as documented and analyzed in
\cite{kacperczyk2018investor}.} The income and wealth distributions in turn
affect other economic phenomena, including the composition of aggregate demand
and the growth rate of aggregate productivity.\footnote{The literature on the
connection between the distribution of income and wealth and growth rates is
extensive. A recent example combining theory and empirics is
\cite{halter2014inequality}.}

One additional advantage of modeling the power law in the firm size
distribution is calibration and testing: matching the Pareto tail index in the
data serves as a valuable additional restriction to fit parameters.

At the same time, analyzing the entry-exit model of \cite{hopenhayn1992entry}
without bounded productivity is nontrivial. One reason is that the lifetime
profits of firms are potentially unbounded, requiring a new approach to the
firm decision problem. Second, the time invariance condition for the
equilibrium measure of firms concerns stationary distributions of Markov
transitions that are possibly transient, due to the unboundedness of
productivity. Third, aggregate output is potentially infinite, since
integration across productivity states is over an unbounded set.

In this paper we provide a comprehensive analysis of the entry-exit model of
\cite{hopenhayn1992entry} without the boundedness assumption. We show that the
complications described above can be cleanly handled by using (i) a weighted
supremum norm for the firm decision problem, as in \cite{boyd1990recursive}
and \cite{alvarez1998dynamic}, and (ii) Kac's Theorem on positive recurrence
to handle productivity dynamics. Through this combination, we provide an exact
necessary and sufficient condition for existence and uniqueness of a
stationary competitive equilibrium in the unbounded setting. In particular, we
show that a stationary competitive equilibrium exists if and only if the
expected lifetime output of firms is finite. We extend the analysis to
monopolistically competitive markets in which each firm's demand bounds its
sales, and show that finite expected lifespan then guarantees existence.

Since expected lifetime output is endogenous, we also provide a sufficient
condition based on a drift restriction over productivity dynamics. Drift
conditions are a well-known technique for controlling Markov processes on
unbounded state spaces (see, e.g., \cite{meyn2012markov}), with the main idea
being to obtain a Lyapunov function on the state space such that (a) the
function becomes large as the state diverges and (b) the value assigned to the
state by the Lyapunov function under the Markov process in question tends to
decrease if the state variable begins to diverge. The main difficulty with the
approach is finding a suitable Lyapunov function. The innovation introduced in
this paper is to use firm output itself as the Lyapunov function. The
resulting drift condition is weak enough to allow a very large range of
specifications for incumbent productivity dynamics. For affine dynamics, which
include Gibrat's law, we prove that it is also necessary.

In addition, under the stated lifetime condition, we provide a decomposition
of the equilibrium firm size distribution and a sample path interpretation via
Pitman's occupation measure. The latter connects the cross-sectional mass of
firms in a given region of the distribution with the occupation times of
individual firms. Using this decomposition, we prove a new formula connecting
aggregate supply (and hence aggregate demand) with the equilibrium entry rate
and the expected lifetime output of firms.\footnote{The proof of existence of
a stationary equilibrium is constructive, so quantitative tractability of the
entry-exit model is preserved.}

With these results in place, the unbounded model becomes the natural setting
in which to study the heavy tails that motivate this paper. We show that, once
the boundedness assumption is removed, the firm size distribution generated by
entry and exit exhibits a Pareto tail under affine productivity dynamics,
including Gibrat's law and the empirically relevant departures from it, and
that the tail index depends only on the law of motion for incumbents and the
technology.\footnote{In a companion paper \citep{shanker2026pareto} we show
that, once the boundedness assumption is removed, the firm size distribution
generated by entry and exit exhibits a Pareto tail under a broad range of
specifications of firm-level dynamics -- including Gibrat's law and the
empirically relevant departures from it -- and that the tail index depends
only on the law of motion for incumbents.} The equilibrium existence theory
developed here is what makes that analysis possible, since it guarantees that
the stationary distribution whose tail is being characterized actually exists.
We show that a stationary competitive equilibrium exists if and only if the
tail index lies strictly above one, that is, above Zipf's law.

Given that Zipf's law marks the boundary of existence of stationary
competitive equilibria, the fact that estimated tail indices lie close to
Zipf's law on both sides
\citep{axtell2001zipf,kondo2023heavy,cao2022firm,chen2026growth} presents a
difficulty. Competitive equilibria are inconsistent with data that appear to
follow Zipf's law. Bounding industry sales does not restore a competitive
equilibrium, since the market could then clear only at a price of zero, at
which price-taking firms make a loss and exit. Nonexistence is only part of
the difficulty. As the tail index approaches Zipf's law from above, mean firm
size in a competitive equilibrium diverges \citep{luttmer2012technology}, and
model quantities change by large amounts for changes in the index smaller than
the differences between published estimates. This limits how well a
competitive model of firm dynamics can pin down its parameters from data
moments.

To obtain stationary equilibria with productivity tails at or below Zipf's
law, we combine market power with a bound on sales that comes from the
consumer demand curve. Under monopolistic competition a stationary equilibrium
then exists whenever expected lifespan is finite. The same Lyapunov inequality
as in the competitive case, with any increasing unbounded function in place of
output, is sufficient for a stationary distribution of firms even when the
tail index of productivity is at or below Zipf's law. In a quantitative
example, we show that the sales of all but the most productive firms follow an
approximate power law whose estimated index can equal one; mean firm
size does not diverge near Zipf's law, while productivity and markups remain
unbounded, and markups inherit the Pareto tail of productivity.

This paper builds on previous studies that have linked random firm-level
growth within an industry to Pareto tails in the cross-sectional distribution
of firm size. Early examples include \cite{champernowne1953model} and
\cite{simon1955class}, who showed that Pareto tails in stationary
distributions can arise if time series follow Gibrat's law along with a
reflecting lower barrier. Since then it has been well understood that Gibrat's
law can generate Pareto tails for the firm size distribution in models where
firm dynamics are exogenously specified. See \cite{gabaix2008power} and
\cite{gabaix2016power} for surveys.

\cite{luttmer2007selection,luttmer2012technology} adds equilibrium entry and
exit decisions and obtains balanced growth paths, but in the specialized
setting of Brownian models of firm growth when the tail index of firm size
exceeds one. Like this paper, \cite{carvalho2019large} study heavy tails in a
Hopenhayn-style entry-exit model with a large but finite number of firms. At
the same time, productivity is still bounded, as in \cite{hopenhayn1992entry}.
\cite{carvalho2019large} also omit conditions under which a stationary
equilibrium exists. In these papers, as in \cite{acemoglu2015innovation} and
\cite{cao2022firm}, the tail rests on Gibrat's law and on particular
specifications of firm dynamics. In \cite{luttmer2011mechanics} growth depends
on a firm type but not on size.

To our knowledge, no earlier existence result for equilibria with endogenous
firm entry and exit covers a general Markov productivity process on an
unbounded state space. Nor does any earlier model with endogenous entry and
exit derive the Pareto tail of firm size when growth can depend on size.
Equilibria with endogenous entry and exit at or below Zipf's law are also, to
our knowledge, new.

%


The remainder of the paper is structured as follows. Section~\ref{s:ee} sets
out the model. Section~\ref{s:sre} shows existence of a unique stationary
competitive equilibrium when the state space is unbounded.
Section~\ref{s:pareto} studies Pareto tails in the competitive equilibrium's
firm size distribution and shows that Zipf's law marks the boundary of
existence. Section~\ref{s:arrangements} shows that a stationary equilibrium
exists at and below the boundary under monopolistic competition, and
Section~\ref{s:c} concludes. Proofs are deferred to the appendix.

\section{Entry and Exit}\label{s:ee}

Apart from unbounded productivity, our baseline assumptions follow
\cite{hopenhayn1992entry}. There is a single good produced by a continuum of
firms, consisting at each point in time of a mixture of new entrants and
incumbents. The good is sold at price $p$ and the demand is given by $D(p)$.

\begin{assumption}\label{a:df}
    The demand function $D$ is continuous and strictly decreasing with $D(0) =
    \infty$ and $\lim_{p\to \infty} D(p) = 0$.
\end{assumption}

Assumption~\ref{a:df} already implies that $p=0$ cannot be an equilibrium,
since demand is infinite at that price. This assumption is convenient but can
be weakened if necessary, since we show below that supply is zero when $p=0$.

Firms facing output price $p$ and having firm-specific productivity $\phi$
generate profits $\pi(\phi, p)$ and produce output $q(\phi, p)$. (We take $q$
and $\pi$ as given but provide examples below where they are derived from
profit maximization problems.) Profits are negative on the boundary due to
fixed costs, as in \cite{hopenhayn1992entry}. In particular,

\begin{assumption}
    \label{a:pq}
    The profit and output functions satisfy the following conditions:
    \begin{enumerate}
        \item The function $\pi \colon \RR^2_+ \to \RR$ is continuous,
            increasing, and strictly increasing in each argument on
            $(0, \infty)^2$.  Moreover, $\pi(\phi, p) < 0$ if either
            $\phi = 0$ or $p = 0$.
        \item The function $q \colon \RR^2_+ \to \RR_+$ is continuous and
            increasing on $\RR^2_+$.  For each $p > 0$, $q(\cdot, p)$ is
            strictly increasing on $(0, \infty)$.
    \end{enumerate}
\end{assumption}

Productivity of each incumbent firm updates according to the idiosyncratic
Markov state process $\Gamma(\phi, \diff \phi')$, where $\Gamma$ is a
transition probability kernel on $\RR_+$. At the beginning of a period, an
incumbent observes current productivity and collects current profit. After
that, the incumbent decides whether to continue into the following period or
exit. If the firm continues, its productivity in the following period is drawn
from $\Gamma(\phi, \cdot)$.\footnote{We take $\Gamma$ to be the law of motion
of productivity relative to a common trend, along a balanced growth path on
which fixed and entry costs grow with the trend. The stationary distribution
studied below is then that of relative productivity, as in
\cite{luttmer2007selection} and \cite{arkolakis2016unified}.} The outside
option for firms is zero and the value $v(\phi, p)$ of an incumbent satisfies
\begin{equation}
    \label{eq:vf}
    v(\phi, p) = \pi(\phi, p) + \beta
        \max \left\{ 0, \int v(\phi', p) \Gamma(\phi, \diff \phi') \right\},
\end{equation}
where $\beta = 1/(1+r)$ for some fixed $r > 0$. Here and below, integrals are
over $\RR_+$.

\begin{assumption}
    \label{a:pgamma}
    The productivity kernel $\Gamma$ is monotone increasing. In addition,
    \begin{enumerate}
        \item For each $a > 0$ and $\phi \geq 0$, there is an $n \in \NN$ such
            that $\Gamma^n(\phi, [0, a)) > 0$.
        \item For each $p > 0$, there exists a $\phi \geq 0$ such that
            $\int \pi(\phi', p) \Gamma(\phi, \diff \phi') \geq 0$.
    \end{enumerate}
\end{assumption}

The symbol $\Gamma^n$ denotes the $n$-step transition kernel. The monotonicity
assumption means that $\Gamma(\phi, [0, a])$ is decreasing in $\phi$ for all
$a \geq 0$. Condition (a) is analogous to the recurrence condition in
\cite{hopenhayn1992entry}. Condition (b) ensures that not all incumbents exit
every period.

New entrants draw productivity independently from a fixed probability
distribution $\gamma$ and enter the market if $\int v( \phi', p) \,
\gamma(\diff \phi') \geq c_e$, where $c_e > 0$ is a fixed cost of entry.

\begin{assumption}
    \label{a:aper}
    The distribution $\gamma$ satisfies $\int q(\phi, p) \gamma(\diff \phi) <
    \infty$ and puts positive mass on the interval $[0, a]$ for all $a > 0$.
\end{assumption}

The first condition in Assumption~\ref{a:aper} is a regularity condition that
helps to ensure finite output. The second condition is convenient because it
leads to aperiodicity of the endogenous productivity process.

For realistic industry dynamics we need both entry and a nonzero rate of exit.
We impose both at a single price $\bar p$.

\begin{assumption}
    \label{a:newent}
    There exists a $\bar p > 0$ such that
    \begin{enumerate}
        \item $\int \pi(\phi, \bar p) \, \gamma(\diff \phi) \geq c_e$, and
        \item
            $\sum_{t \geq 1} \beta^t \int \pi(\phi', \bar p) \, \allowbreak
            \Gamma^t(0, \diff \phi') < 0$.
    \end{enumerate}
\end{assumption}

Condition (a) in Assumption~\ref{a:newent} ensures that entry occurs when the
price is sufficiently large. It is relatively trivial because, for price
taking firms, revenue is proportional to price. Condition (b) states that, at
the same price, a firm currently at zero productivity has strictly negative
expected discounted profits from next period onward if it never exits.

One setting where Assumption~\ref{a:newent}(b) holds is when firm growth
follows Gibrat's law, so that $\Gamma$ is represented by the recursion
$\phi_{t+1} = A_{t+1} \phi_t$ for some positive {\sc iid} sequence $\{A_t\}$.
Then $\phi_0=0$ implies $\phi_t=0$ for all $t$, and hence $\Gamma^t (0, \diff
\phi')$ is a point mass at zero. Hence the integral in
Assumption~\ref{a:newent}(b) evaluates to $\pi(0, p)$ for each $t$, which is
negative by Assumption~\ref{a:pq}. Assumption~\ref{a:newent}(b) also holds in
cases where zero productivity is not absorbing. Example~\ref{eg:gll} verifies
Assumption~\ref{a:newent}(b) at all prices below a bound set by the
primitives.

Since productivity is unbounded and profits can be arbitrarily large, we also
need a condition on the primitives to ensure that $v$ is finite. In stating
it, we consider the productivity process $\{\phi_t\}$ defined by
\begin{equation}
    \label{eq:inff}
        \phi_0 \sim \gamma \text{ and }
        \phi_{t+1} \sim \Gamma(\phi_t, \diff \phi')
        \text{ when } t \geq 0.
\end{equation}

\begin{assumption}
    \label{a:fc}
    There is a $\delta \in (\beta, 1)$ with $\sum_{t \geq 0} \, \delta^t \,
    \EE \, \pi (\phi_t, p) < \infty$ at all $p \geq 0$.
\end{assumption}

While slightly stricter than a direct bound on lifetime profits,
Assumption~\ref{a:fc} has the benefit of yielding a contraction result for the
Bellman operator corresponding to the Bellman equation \eqref{eq:vf}. Since we
are working in a setting where profits can be arbitrarily large, the value
function is unbounded, so the contraction in question must be with respect to
a \emph{weighted} supremum norm. To construct this norm, we take $\delta$ as
in Assumption~\ref{a:fc} and let
\begin{equation}
    \label{eq:dkappa}
    \kappa(\phi, p) := \sum_{t \geq 0}
        \delta^t \EE_\phi \hat \pi (\phi_t, p)
        \; \text{ with }  \hat \pi := \pi + b.
\end{equation}
Here $\EE_\phi$ denotes expectation for the process $\{\phi_t\}$ started at
$\phi_0 = \phi$, and $b$ is a constant chosen such that $\pi + b \geq 1$. The
function $\kappa$ is constructed so that it dominates the value function in
absolute value up to a constant multiple and satisfies $1 \leq \kappa <
\infty$ at all points in the state space.\footnote{To be more precise, $\phi
\mapsto \kappa(\phi, p)$ is finite $\gamma$-almost everywhere by
Assumption~\ref{a:fc}. If $\gamma$ is supported on all of $\RR_+$, then, since
the function in question is monotone, this implies that $\kappa$ is finite
everywhere. If not, then we tighten the assumptions above by requiring that
$\kappa(\phi, p)$ is finite everywhere.} For each scalar-valued $f$ on
$\RR_+^2$, let $\|f \|_\kappa := \sup |f/\kappa|$. This is the
$\kappa$-weighted supremum norm. If it is finite for $f$ then we say that $f$
is $\kappa$-bounded. Let
\begin{equation*}
    \cf :=
    \text{all continuous, increasing and $\kappa$-bounded functions on
    $\RR_+^2$}.
\end{equation*}
Under the distance $d(v, w) := \| w - v \|_\kappa$, the set $\cf$ is a
complete metric space.\footnote{Completeness of the set of continuous
$\kappa$-bounded functions under $d$ is proved in many places, including
\cite{hernandez2012further}, \S7.2. Our claim of completeness of $(\cf, d)$
follows from the fact that the limit of a sequence of increasing functions in
$(\cf, d)$ is also increasing.}


\begin{assumption}
    \label{a:gcm}
    If $u$ is in $\cf$, then $(\phi, p) \mapsto \int u(\phi', p) \Gamma(\phi,
    \diff \phi')$ is continuous.
\end{assumption}

Assumption~\ref{a:gcm} is a version of the continuity property imposed by
\cite{hopenhayn1992entry}, modified slightly to accommodate the fact that
$\Gamma$ acts on unbounded functions.

\section{Stationary Competitive Equilibrium}

\label{s:sre}

Now we turn to existence, uniqueness and computation of stationary competitive
equilibria for the industry. The equilibrium concept is that of
\cite{hopenhayn1992entry}. All assumptions from the previous section are in
force.

\subsection{Preliminary Results}

We begin our analysis with the firm decision problem. The next lemma
determines the firm value function $\bar v$, where $\bar v(\phi, p)$ is
lifetime value of the firm given current productivity $\phi$ and price $p$.

\begin{lemma}
    \label{l:tcm}
    The Bellman operator $T \colon \cf \to \cf$ defined at $v \in \cf$ by
    \begin{equation}
        \label{eq:bellop}
        (Tv)(\phi, p) = \pi(\phi, p) + \beta
            \max \left\{ 0, \int v(\phi', p) \Gamma(\phi, \diff \phi')
            \right\}
    \end{equation}
    is a contraction of modulus $\beta/\delta$ on the metric space $(\cf, d)$.
    Its unique fixed point $\bar v$ is increasing, strictly increasing in each
    argument on $(0, \infty)^2$, and $\bar v(\phi, 0) = \pi(\phi, 0) < 0$ for
    all $\phi \geq 0$.
\end{lemma}

Given $\bar v$, we let $\bar \phi$ be the \emph{exit threshold} function
defined by
\begin{equation}
    \label{eq:barphi}
    \bar \phi(p) :=
    \min
    \left\{
        \phi \geq 0
        \;\; \big| \;
        \int \bar v(\phi', p) \, \Gamma(\phi, \diff \phi') \geq 0
    \right\}.
\end{equation}
With the convention that incumbents who are indifferent remain rather than
exit, an incumbent with productivity $\phi$ exits if and only if $\phi < \bar
\phi (p)$. In \eqref{eq:barphi} we take the usual convention that $\min
\varnothing = \infty$.

\begin{lemma}
    \label{l:barphi}
    The exit threshold $\bar \phi$ is finite and decreasing on $(0, \infty)$
    with $\bar \phi(0) = \infty$. Moreover, for every $p > 0$ and every $\phi
    \geq 0$, $\int \bar v(\phi', p) \, \Gamma(\phi, \diff \phi') \geq 0$ if
    and only if $\phi \geq \bar \phi(p)$.
\end{lemma}

\subsection{Definitions}\label{ss:sredef}

Let $\bB$ be the Borel subsets of $\RR_+$ and $\mM$ be all measures on $\bB$.
Taking $\bar v$ and $\bar \phi$ as defined in the previous section, a
\emph{stationary equilibrium} is a triple
\begin{equation*}
    (p, M, \mu) \text{ in }
    \eE \coloneq (0, \infty) \times (0, \infty) \times \mM,
\end{equation*}
with $p$ understood as price, $M$ as mass of entrants, and $\mu$ as a
distribution of firms over productivity levels, such that the goods market
clears:
\begin{equation}
    \label{eq:goods}
    \int q( \phi, p) \mu(\diff \phi) = D(p),
\end{equation}
the \emph{invariance condition}
\begin{equation}\label{eq:invar}
    \mu (B)
    = \int \Gamma(\phi, B) \, \1\{\phi \geq \bar \phi(p)\} \,  \mu(\diff \phi)
    + M \, \gamma(B)
    \text{ for all } B \in \bB,
\end{equation}
holds, the \emph{equilibrium entry condition}
\begin{equation}\label{eq:entc}
    \int \bar v( \phi, p) \, \gamma(\diff \phi) = c_e
\end{equation}
holds, and the \emph{balanced entry and exit condition}
\begin{equation}\label{eq:beec}
    M = \mu \{ \phi < \bar \phi(p) \}
\end{equation}
is verified.

\subsection{Existence and Uniqueness}

Throughout this section, we take $\{\phi_t\}$ as in \eqref{eq:inff} and set
\begin{equation}\label{eq:taustar}
    \tau(p) := \inf \setntn{t \geq 0}{\phi_t < \bar \phi(p)}.
\end{equation}
The random variable $\tau(p)$ records \emph{firm lifespan} associated with
productivity path $\{\phi_t\}$ when output price is $p$.\footnote{We call
$\tau(p)$ lifespan although the firm still produces at date $\tau(p)$, the
date it exits. A firm whose initial draw $\phi_0$ lies below $\bar \phi(p)$
has $\tau(p)=0$ and produces at date 0 only.} We set $\inf \varnothing =
\infty$. \emph{Lifetime firm output} is the random variable
\begin{equation}\label{eq:ellstar}
    \ell(p) := \sum_{t=0}^{\tau(p)} q(\phi_t, p).
\end{equation}

The first result of this section provides a candidate equilibrium price, which
equates the expected value of entry to its cost.

\begin{lemma}
    \label{l:pphis}
    There exists a unique $p > 0$ such that $\int \bar v( \phi, p) \,
    \gamma(\diff \phi) = c_e$.
\end{lemma}

In what follows we let
\begin{equation}\label{eq:eqenp}
    p^* \coloneq
    \text{ the unique positive price in Lemma~\ref{l:pphis} }
\end{equation}
and call it the \emph{equilibrium entry price}.

\begin{corollary}
    \label{l:pintv}
    If $\bar p$ satisfies Assumption~\ref{a:newent}, then $\bar \phi(p) > 0$
    for every $p \in (0, \bar p]$. Moreover, $p^* \leq \bar p$, and hence
    $\bar \phi(p^*) > 0$.
\end{corollary}

The proof in the appendix proceeds by contradiction. The intuition is as
follows. Suppose $\bar \phi(\bar p) = 0$. By Lemma~\ref{l:barphi}, a firm at
zero productivity then continues, so $(\Gamma\bar v)(0, \bar p) \geq 0$. Since
$\Gamma\bar v$ is increasing, the firm continues at every productivity state
and never exits, so its discounted continuation value at zero productivity is
the sum in Assumption~\ref{a:newent}(b). That sum is negative, so $(\Gamma\bar
v)(0, \bar p) < 0$, a contradiction.

\begin{remark}
    \label{r:converse}
    Condition (b) of Assumption~\ref{a:newent} is not only sufficient for a
    positive exit threshold at $\bar p$ but also necessary. To see this, fix
    $p > 0$ and let
    \begin{equation*}
        S(p) := \sum_{t \geq 1} \beta^t
            \int \pi(\phi', p) \, \Gamma^t(0, \diff \phi').
    \end{equation*}
    Never exiting is a feasible policy, so $\beta (\Gamma \bar v)(0, p) \geq
    S(p)$, and the argument in the proof of Corollary~\ref{l:pintv} shows that
    equality holds when $\bar \phi(p) = 0$. Hence $\bar \phi(p) > 0$ if and
    only if $S(p) < 0$. In particular, condition (b) holds at $\bar p$ if and
    only if $\bar \phi(\bar p) > 0$. The appendix gives details.
\end{remark}

We can now state our main existence and uniqueness result, which characterizes
equilibrium for the entry-exit model set out in Section~\ref{s:ee}. All
assumptions from that section are in force.

\begin{theorem}
    \label{t:bk1}
    The following statements are equivalent:
    \begin{enumerate}
        \item[\rm{(a)}] $\EE \ell(p^*) < \infty$.
        \item[\rm{(b)}] There exists an $M^* \in (0, \infty)$ and $\mu^*$ in
            $\mM$ such that  $(p^*, M^*, \mu^*)$ is a stationary equilibrium.
    \end{enumerate}
    If either and hence both of these statements are true, then
    \begin{enumerate}
        \item[\rm{(i)}] $(p^*, M^*, \mu^*)$ is the only stationary equilibrium
            in $\eE$,
        \item[\rm{(ii)}] equilibrium expected firm lifespan $\EE \, \tau(p^*)$
            is finite,
        \item[\rm{(iii)}] the equilibrium $(p^*, M^*, \mu^*)$ obeys
            \begin{equation}
                \label{eq:kdec}
                \mu^* (B) = M^*  \cdot
                \EE \sum_{t = 0}^{\tau(p^*)} \1\{\phi_t \in B\}
                    \quad \text{for all } B \in \bB,
            \end{equation}
        \item[\rm{(iv)}] and aggregate supply obeys
            \begin{equation}\label{eq:ase}
                \int q(\phi, p^*) \mu^* (\diff \phi)
                    =  M^* \EE \, \ell(p^*).
            \end{equation}
    \end{enumerate}
\end{theorem}

The decomposition \eqref{eq:kdec} ties the cross-sectional distribution of
productivity to dynamics at the level of the firm. It says that the mass of
firms in set $B$ is proportional to the expected number of times that a firm's
productivity visits $B$ over its lifespan. The decomposition is obtained by a
combination of Kac's Theorem and the Pitman occupation formula. One simple
special case of \eqref{eq:kdec} is when $B = \RR_+$, which yields
\begin{equation*}
    \frac{M^*}{\mu^*(\RR_+)} = \frac{1}{\EE \, [\tau(p^*)+1]}.
\end{equation*}
Thus, in equilibrium, the entry rate equals the reciprocal of the expected
number of periods firms produce. Finite expected lifespan, condition (ii), is
therefore what makes the population of firms finite. Market clearing in (iv)
requires the stronger condition (a), finite expected lifetime output, and
Section~\ref{s:arrangements} shows that the first condition can hold where the
second fails.

The result in \eqref{eq:ase} provides a simple formula connecting aggregate
supply (and aggregate demand) with the equilibrium entry rate and the expected
lifetime output of firms. The formula may be used for calibration or testing
in quantitative work. The proof of Theorem~\ref{t:bk1} is constructive, and
Appendix~\ref{s:comp} turns it into an algorithm for computing the
equilibrium.

One special case of Theorem~\ref{t:bk1} is when productivity is bounded above,
as in \cite{hopenhayn1992entry}. To see this, suppose $\phi_t$ is bounded
above by $\phi_{\max}$ and let $\mathbf a^* \coloneq [0, \bar \phi(p^*))$,
which is nonempty by Corollary~\ref{l:pintv}. By Assumption~\ref{a:pgamma}
there exists an integer $n$ such that $\epsilon := \Gamma^n (\phi_{\max},
\mathbf a^*) > 0$. By monotonicity of $\Gamma$, for any $\phi \leq
\phi_{\max}$ we have $\Gamma^n(\phi, \mathbf a^*) \geq \epsilon$. Hence a firm
that has not exited by date $m$ exits by date $m+n$ with probability at least
$\epsilon$, whatever its history. It follows that $\PP\{\tau(p^*) > m\} \leq
(1-\epsilon)^{\lfloor m/n \rfloor}$ and
\begin{equation*}
    \EE \, [\tau(p^*)+1]
    = 1 + \sum_{m \geq 0} \PP\{\tau(p^*) > m\}
    \leq 1 + \sum_{m \geq 0} (1-\epsilon)^{\lfloor m/n \rfloor}
    < \infty.
\end{equation*}
Since $q(\phi_t, p^*) \leq \bar q \coloneq q(\phi_{\max}, p^*)$, this implies
$\EE \, \ell(p^*) \leq \bar q \EE \, [\tau(p^*)+1] < \infty$. In particular,
bounded productivity implies (a) in Theorem~\ref{t:bk1}, and hence (b) and
(i)--(iv).

A consequence of market clearing is that Zipf's law cannot hold exactly in a
stationary equilibrium. Section~\ref{s:pareto} develops this point.

\subsection{Restrictions on Primitives}\label{ss:ascond}

Theorem~\ref{t:bk1} characterizes existence and uniqueness of a stationary
equilibrium in terms of expected lifetime output, which depends on the
endogenous price $p^*$ and exit threshold $\bar \phi(p^*)$. In this section we
give restrictions on primitives. Proposition~\ref{p:geoerg} shows that a drift
condition, with output as the Lyapunov function, is sufficient, and
Proposition~\ref{p:affnec} shows that, for affine productivity dynamics, such
drift is also necessary.

\begin{assumption}
    \label{a:unbq}
    For each $p > 0$, output obeys $q(\phi, p) \to \infty$ as $\phi \to
    \infty$.
\end{assumption}

Assumption~\ref{a:unbq} says that a firm's output grows without bound as its
productivity does. It holds in the examples below and in any specification
where output is a power of productivity. Its role is technical: it ensures
that the set of productivity levels at which output is below any given level
is bounded.

\begin{definition}
    \label{d:fl}
    Productivity growth obeys the \emph{Foster--Lyapunov condition} if, for
    each $p > 0$, there exist $\lambda \in (0, 1)$ and $L < \infty$ such that
    \begin{equation}
        \label{eq:dc}
        \int q(\phi', p) \Gamma(\phi, \diff \phi')
        \leq \lambda q(\phi, p) + L
        \quad
        \text{ for all } \phi \geq 0.
    \end{equation}
\end{definition}

The Foster--Lyapunov condition says that output growth for incumbents is
expected to be negative whenever current output is sufficiently
large.\footnote{To see this, we can write $q(\phi_t, p)$ as $Q_t$ and express
\eqref{eq:dc} as $\ln( \EE_t Q_{t+1} / Q_t)\leq \ln (\lambda + L / Q_t)$. When
$Q_t$ is sufficiently large, the right-hand side is negative.} In the
literature on Markov processes, bounds of the form \eqref{eq:dc} are called
Foster--Lyapunov drift conditions. The key idea is that the output function
$q$ can be adopted as the Lyapunov function.

Let the assumptions in Section~\ref{s:ee} hold. The following result shows
that the Foster--Lyapunov condition is sufficient for existence and uniqueness
of a stationary equilibrium.

\begin{proposition}
    \label{p:geoerg}
    Let Assumption~\ref{a:unbq} hold. If productivity growth obeys the
    Foster--Lyapunov condition, then the model has a unique stationary
    equilibrium.
\end{proposition}

The intuition behind Proposition~\ref{p:geoerg} is as follows. When
productivity growth obeys the Foster--Lyapunov condition, incumbents with
sufficiently large output tend to see output fall in the next period. Output
is a strictly increasing function of $\phi$, so falling output means falling
productivity. From this one can construct a finite interval such that, for any
given incumbent, productivity returns to this interval infinitely often if the
incumbent never exits. At each such occasion, the recurrence condition in
Assumption~\ref{a:pgamma}, together with monotonicity of $\Gamma$, yields a
probability of at least $\epsilon > 0$ of exiting within a fixed number of
periods. Eventually the firm exits and lifetime output remains
finite.\footnote{Even if firm lifespan is finite along every sample path, this
does not necessarily imply that the expectation $\EE \tau(p^*)$ is finite.
Hence there are some subtleties involved in the proof of
Proposition~\ref{p:geoerg}. The reason that output is used as the Lyapunov
function is that we need this expectation to be finite. The appendix gives
details.}

\begin{example}
    \label{eg:gll}
    Suppose that incumbent productivity grows according to
    \begin{equation}
        \phi_{t+1} = A_{t+1} \phi_t + Y_{t+1}
        \quad \text{for some nonnegative {\sc iid} sequence } \{A_t, Y_t\},
    \end{equation}
    that production is linear in $\phi$ and that all factors of production are
    constant, so that $q(\phi, p) = e \phi$ for some $e > 0$. Profits are
    $\pi(\phi, p) = p e \phi - c_f$, where $c_f > 0$ is the fixed cost.
    Regarding the Foster--Lyapunov condition \eqref{eq:dc}, we have
    \begin{equation*}
        \int q(\phi', p) \Gamma(\phi, \diff \phi')
        = e \EE A_{t+1} \phi + e \, \EE Y_{t+1}
        = \EE A_{t+1} q(\phi, p) + e \, \EE Y_{t+1}.
    \end{equation*}
    Productivity growth therefore obeys the Foster--Lyapunov condition
    whenever $\EE A_t < 1$ and $\EE Y_t < \infty$. Since $q(\phi, p) = e
    \phi$, Assumption~\ref{a:unbq} also holds.

    When $Y_t \equiv 0$, growth follows Gibrat's law, so
    Assumption~\ref{a:newent}(b) holds at every price by the argument in
    Section~\ref{s:ee}. Now suppose instead that $\EE Y > 0$, so a firm can
    recover from a low productivity shock and zero productivity is not
    absorbing. In this case, Assumption~\ref{a:newent}(b) restricts possible
    equilibrium prices. Provided $\beta \EE A < 1$ and $\EE Y < \infty$,
    iterating the recursion from zero and summing discounted profits gives
    \begin{equation*}
        \sum_{t \geq 1} \beta^t
        \int \pi(\phi', p) \, \Gamma^t(0, \diff \phi')
        = \frac{\beta}{1 - \beta}
        \left(\frac{p e \, \EE Y}{1 - \beta \EE A} - c_f\right).
    \end{equation*}
    {\sloppy Hence Assumption~\ref{a:newent}(b) bounds $\bar p$ above, holding
    if and only if $\bar p < c_f(1 - \beta \EE A)/(e \, \EE Y)$. By
    Remark~\ref{r:converse}, the same calculation shows that $\bar \phi(p) =
    0$ for every $p \geq c_f(1 - \beta \EE A)/(e \, \EE Y)$: at such prices a
    firm at zero productivity expects to recover and never exits.\par}

    The expected profit of an entrant is $\int \pi(\phi, \bar p) \,
    \gamma(\diff \phi) = \bar p e \, \EE_\gamma \phi - c_f$, where $\EE_\gamma
    \phi := \int \phi \, \gamma(\diff \phi)$. By Assumption~\ref{a:aper}, the
    integral $\int q(\phi, p) \, \gamma(\diff \phi) = e \, \EE_\gamma \phi$ is
    finite, and hence so is $\EE_\gamma \phi$. Assumption~\ref{a:newent}(a)
    therefore bounds $\bar p$ below, holding if and only if $\EE_\gamma \phi >
    0$ and $\bar p \geq (c_e + c_f)/(e \, \EE_\gamma \phi)$. Conditions (a)
    and (b) of Assumption~\ref{a:newent} therefore both hold at some $\bar p$
    if and only if
    \begin{equation*}
        \EE_\gamma \phi
        \;>\;
        \EE Y \cdot \frac{c_e + c_f}{c_f (1 - \beta \EE A)}.
    \end{equation*}
\end{example}

\begin{example}
    \label{eg:glcd}
    Suppose instead that production is Cobb--Douglas, with output $\phi
    n^\theta$ under labor input $n$ and parameter $\theta \in (0, 1)$. With
    profits given by $p \phi n^\theta - c_f - w n$ for some $c_f, w > 0$, the
    function for output at optimal labor input is
    \begin{equation*}
         q(\phi, p) = \phi^\eta m(p)
         \quad \text{where }
         \eta := \frac{1}{1-\theta}
         \; \text{ and } \;
         m(p) := \left( \frac{p \theta}{w} \right)^{\theta/(1-\theta)}.
    \end{equation*}
    If productivity growth follows $\phi_{t+1} = A_{t+1} \phi_t$, then the
    left-hand side of the Foster--Lyapunov condition \eqref{eq:dc} becomes
    \begin{equation*}
        \int q(\phi', p) \Gamma(\phi, \diff \phi')
        = \EE  (A_{t+1} \phi)^\eta \, m(p)
        =  \EE  A_{t+1}^\eta  q(\phi, p).
    \end{equation*}
    Thus, productivity growth obeys the Foster--Lyapunov condition whenever
    $\EE[ A_t^\eta ] < 1$, and Assumption~\ref{a:unbq} holds since $m(p) > 0$
    for $p > 0$. If, say, $A_t$ is lognormal with $A_t = \exp(m + \sigma Z)$
    for $Z \sim N(0,1)$, then $\EE[ A_t^\eta ] = \EE \exp(\eta m + \eta \sigma
    Z) = \exp(\eta m + (\eta \sigma)^2 / 2)$ and the condition becomes
    \begin{equation*}
        m + \frac{1}{1-\theta} \frac{\sigma^2}{2} < 0.
    \end{equation*}
    This joint restriction on the rate of incumbent firm growth and the
    Cobb--Douglas production parameter $\theta$ is sufficient for the
    Foster--Lyapunov condition, and hence, by Proposition~\ref{p:geoerg}, for
    existence of a unique stationary equilibrium.

    Under Gibrat growth, zero is absorbing and $\pi(0,p)=-c_f$. Hence the sum
    in Assumption~\ref{a:newent}(b) is $-\beta c_f/(1-\beta)<0$ at every
    price. At the optimal labor choice, $w n=\theta p q(\phi,p)$, so
    \begin{equation*}
        \int \pi(\phi,p)\,\gamma(\diff\phi)
        =(1-\theta)\left(\frac{\theta}{w}\right)^{\theta/(1-\theta)}
        p^\eta \EE_\gamma\phi^\eta-c_f.
    \end{equation*}
    The moment $\EE_\gamma\phi^\eta :=\int\phi^\eta\gamma(\diff\phi)$ is
    finite by Assumption~\ref{a:aper}. Assumption~\ref{a:newent}(a) holds at
    some price if and only if $\EE_\gamma\phi^\eta>0$.
\end{example}

Within the classes of specifications in Examples~\ref{eg:gll}
and~\ref{eg:glcd}, the Foster--Lyapunov condition is also necessary for
existence: if it fails, expected lifetime output is infinite, and no
stationary equilibrium exists. The next result proves this for a class
containing both examples: affine productivity dynamics, which include Gibrat's
law as the case $Y \equiv 0$, combined with output that is a power of
productivity.

\begin{assumption}
    \label{a:affine}
    The primitives satisfy the following conditions:
    \begin{enumerate}
        \item Incumbent productivity evolves according to
            $\phi_{t+1} = A_{t+1} \phi_t + Y_{t+1}$, where the pairs
            $(A_t, Y_t)$, $t \geq 1$, are {\sc iid} copies of the nonnegative
            random vector $(A, Y)$.
        \item Output is $q(\phi, p) = m(p) \phi^\eta$ for some $\eta > 0$ and
            some function $m$ with $m(p) > 0$ for all $p > 0$.
        \item The entry distribution satisfies $\gamma((x, \infty)) > 0$ for
            every $x \geq 0$.
    \end{enumerate}
\end{assumption}

\begin{proposition}
    \label{p:affnec}
    If Assumption~\ref{a:affine} holds, then the following statements are
    equivalent:
    \begin{enumerate}
        \item[\rm{(a)}] The model has a unique stationary equilibrium
            $(p^*, M^*, \mu^*)$.
        \item[\rm{(b)}] Productivity growth obeys the Foster--Lyapunov
            condition.
        \item[\rm{(c)}] $\EE A^\eta < 1$ and $\EE Y^\eta < \infty$.
    \end{enumerate}
\end{proposition}

Example~\ref{eg:gll} is the case $\eta = 1$ and $m \equiv e$, and the Gibrat
specification in Example~\ref{eg:glcd} is the case $Y \equiv 0$ and $\eta =
1/(1-\theta)$. In the Gibrat case, the necessity part of
Proposition~\ref{p:affnec} has a sharp reading. If $\EE \ln A < 0$, the log of
productivity is a random walk with negative drift, so every firm exits with
probability one and, under the conditions of Theorem~\ref{t:zipf}, with finite
expected lifespan (Proposition~\ref{l:life_drift}). If, in addition, $\EE
A^\eta \geq 1$, then expected lifetime output is nonetheless infinite, and, by
Theorem~\ref{t:bk1}, no stationary equilibrium exists. The reason is that the
rare firms that survive for long periods grow fast enough to dominate expected
output. Bounded productivity rules out this possibility by assumption.

\section{Pareto Tails and Zipf's Law}\label{s:pareto}

We now turn to the shape of the competitive equilibrium firm size
distribution. Measure firm size by output, so that the equilibrium firm size
distribution is the image $\nu^*$ of $\mu^*$ under $\phi \mapsto q(\phi,
p^*)$, defined by $\nu^*(B) := \mu^*\setntn{\phi \in \RR_+}{q(\phi, p^*) \in
B}$ for $B \in \bB$. As usual, a finite measure $\nu$ on $\bB$ is said to have
a \emph{Pareto tail} with \emph{tail index} $\kappa > 0$ if $\lim_{x \to
\infty} x^\kappa \, \nu(x, \infty)$ exists and lies in $(0, \infty)$. Market
clearing gives $\int x \, \nu^*(\diff x) = D(p^*) < \infty$, while a Pareto
tail with index $\kappa \leq 1$ has infinite mean. Hence no stationary
competitive equilibrium has a firm size distribution with a Pareto tail of
index one or less; in particular, Zipf's law cannot hold exactly.\footnote{The
requirement that firm size cannot be Zipf in equilibrium also appears in 
more restricted settings than ours, including Brownian models of firm growth
\citep{luttmer2007selection,luttmer2012technology} 
and static trade models \citep{chaney2008distorted}.} 
The rest of this section shows that with unbounded productivity Pareto tails can
arise in equilibrium, with indices arbitrarily close to one.

Under Assumption~\ref{a:affine}, the equilibrium firm size distribution
$\nu^*$ has a Pareto tail, and the condition for existence of equilibrium can
be stated in terms of its tail index.

\begin{theorem}
    \label{t:zipf}
    Let Assumption~\ref{a:affine} hold. Suppose in addition that there exist
    $\alpha > 0$ and $\delta > 0$ such that $\EE A^\alpha = 1$ and $\EE
    A^{\alpha + \delta} < \infty$, that the distribution of $\ln A$
    conditional on $A > 0$ is nonarithmetic, and that $\EE Y^\alpha$ and $\int
    z^\alpha \gamma(\diff z)$ are finite. Let $\kappa := \alpha / \eta$.
    \begin{enumerate}
        \item[\rm{(a)}] If $\kappa > 1$, then the model has a unique
            stationary equilibrium $(p^*, M^*, \mu^*)$, and its firm size
            distribution $\nu^*$ has a Pareto tail with tail index $\kappa$.
        \item[\rm{(b)}] If $\kappa \leq 1$, then the model has no stationary
            equilibrium.
    \end{enumerate}
\end{theorem}

\begin{figure}[t]
	\centering
	\scalebox{0.72}{\includegraphics{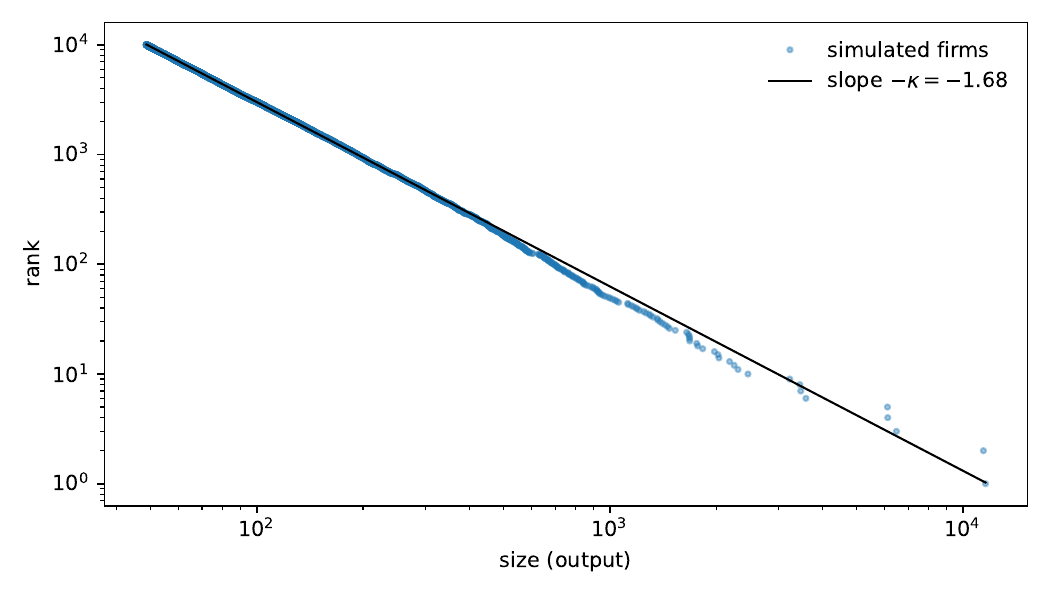}}
	\caption{\label{f:fds_pareto}Rank-size plot of the equilibrium firm size
    distribution, with a reference line of slope $-\kappa$}
\end{figure}

What Theorem~\ref{t:zipf} adds is that the primitives pin down the tail index.
The exponent $\alpha$ is the familiar Kesten--Goldie exponent of the
multiplicative component of incumbent growth. It does not depend on the entry
distribution, the exit threshold or demand, provided entrants have lighter
tails than incumbents. Since firm size is $q = m(p) \phi^\eta$, the primitives
determine the candidate tail index $\kappa = \alpha / \eta$ of firm size.

Figure~\ref{f:fds_pareto} illustrates Theorem~\ref{t:zipf} in the setting of
Example~\ref{eg:glcd}, with $\theta = 0.3$, so that $\eta \approx 1.43$, and
with $\ln A$ normal with mean $-0.012$ and standard deviation $0.1$, so that
$\alpha = 2.4$ and $\kappa = 1.68$. The remaining parameters and the
computational method are given in Appendix~\ref{s:comp}. The figure plots rank
against size, both on log scales, for the largest one percent of firms in a
simulated cross section of one million firms drawn from the equilibrium firm
size distribution. Under a Pareto tail with index $\kappa$, these points
should lie close to a line of slope $-\kappa$, which is also shown.

\begin{table}[t]
\centering
\small
\begin{threeparttable}
\begin{tabular*}{\textwidth}{@{\extracolsep{\fill}} S[table-format=1.3] S[table-format=1.3] S[table-format=2.1] S[table-format=3.1] S[table-format=1.2] S[table-format=2.1] S[table-format=1.2] @{}}
\toprule
{$\kappa$} & {$p^*$} & {Expected lifespan} & {Mean size} & {Entry flow} & {Large-firm share} & {Sales index} \\
\midrule
1.50  & 1.361 & 7.1  & 3.2    & 3.89 & 33.2 & 1.50 \\
1.06  & 1.312 & 10.0 & 18.8   & 0.49 & 87.6 & 1.06 \\
1.03  & 1.308 & 10.3 & 36.5   & 0.24 & 93.6 & 1.03 \\
1.01  & 1.306 & 10.5 & 107.3  & 0.08 & 97.8 & 1.01 \\
\bottomrule
\end{tabular*}
\begin{tablenotes}[flushleft]
\footnotesize
\item \textit{Note}: Size is sales in units of an entrant's mean sales.
The entry flow is the sales of entrants per hundred units of industry
sales, and the large-firm share is the percentage of sales from firms
above ten times an entrant's size. Following \citet{chen2026growth},
the sales index is
$\ln[N(\phi_1)/N(\phi_2)] / \ln[\sS(\phi_2)/\sS(\phi_1)]$, where $N(\phi)$ is
the
mass of firms with productivity above $\phi$ and $\sS(\phi)$ is the sales of a
firm with productivity $\phi$, with $\phi_1 = 10$ and $\phi_2 = 100$.
See Appendix~\ref{s:comp}.
\end{tablenotes}
\end{threeparttable}
\medskip
\caption{\label{tab:boundary_main}Simulated competitive equilibrium as
  $\kappa$ falls to one}
\end{table}

A competitive equilibrium also predicts high market concentration close to
Zipf's law, where tail indices of firm size estimated in the data
lie.\footnote{\cite{axtell2001zipf} estimates 1.06 for employment and 0.99 for
receipts in 1997, \cite{digiovanni2011power} 1.02 for the sales of French
firms in 2006, and \cite{kondo2023heavy} between 0.76 and 1.23 for size
thresholds from 2 to 10,000 employees. For the top one percent of firms by
employment in 28 states, \cite{cao2022firm} find 1.17 in 1995 and 0.99 in
2014, and \cite{chen2026growth} finds a decline over 1978 to 2019 in Census
data.} Table~\ref{tab:boundary_main} shows market concentration in the model
of Figure~\ref{f:fds_pareto} as $\kappa$ approaches one from above. Mean firm
size rises, the entry flow and the mass of firms fall toward zero, and output
concentrates in a diminishing measure of large firms. Market concentration is
also highly sensitive to the index near the boundary; between $\kappa = 1.06$
and $\kappa = 1.01$ the mean firm grows from 19 to 107 times an entrant and
the large-firm share from 88 to 98
percent.\footnote{\cite{luttmer2012technology} likewise finds entry and exit
flows sensitive to small changes in the tail index, but does not define equilibrium
at or below Zipf's law.} Despite the model
predictions, estimated tail indices are close to or below one, and
concentration observed in the data is stable and well below what
Table~\ref{tab:boundary_main} predicts near the boundary.\footnote{In the
United States, the average firm is three to four times the size of an entrant,
and the share of employment in firms with at least 250 employees rose only
from 51 to 57 percent between 1977 and 2014 \citep{hopenhayn2022population}.
Firms with at least 50 employees, about ten times an entrant, held 69 percent
of employment in 1997 \citep{kondo2023heavy}, and the largest 0.02 percent of
firms held about a quarter of employment over 1977 to 2012
\citep{carvalho2019large}.}

\section{Stationary Equilibrium at and Below Zipf's Law}\label{s:arrangements}

Under bounded productivity, as in \cite{hopenhayn1992entry}, non-existence of
stationary competitive equilibria at or below Zipf's law does not arise, since
firm size is bounded and expected lifetime output is finite. Bounding
productivity, by a compact support or by truncating a Pareto distribution, has
a cost. At or below Zipf's law every aggregate then depends on the bound,
which is a parameter with no economic content or stable empirical counterpart.
As the bound goes to infinity, mean firm size once again diverges and the mass
of firms diminishes to zero. If instead we bounded aggregate sales under price
taking, the market would clear only at a price of zero, where every firm would
make a loss and exit.

We now turn to stationary equilibria in which the tail index of productivity
can be at or below Zipf's law. In particular, we study monopolistically
competitive equilibria, in which each firm's demand bounds its sales, while
productivity remains unbounded and its tail goes into unbounded
markups.
Our first step is to extend the results above to
show that even if the mean of
productivity is infinite, a stationary distribution of firms exists at an
entry price $p^*$ of \eqref{eq:eqenp} whenever expected lifespan is finite.


\begin{proposition}
    \label{p:demog}
    Let $\Lambda_p(B) := \EE \sum_{t=0}^{\tau(p)} \1\{\phi_t \in B\}$ for $B
    \in \bB$, and fix a mass of entrants $M > 0$.
    \begin{enumerate}
        \item[\rm{(a)}] A finite measure $\mu$ on $\bB$ such that
            $(p^*, M, \mu)$ satisfies \eqref{eq:invar} and \eqref{eq:beec}
            exists if and only if $\EE \tau(p^*) < \infty$.
        \item[\rm{(b)}] If $\EE \tau(p^*) < \infty$, then
            $\mu = M \Lambda_{p^*}$ is the only such measure and
            $(p^*, M, \mu)$ satisfies \eqref{eq:entc}.  The mass of firms is
            $\mu(\RR_+) = M \EE[\tau(p^*) + 1]$.
    \end{enumerate}
\end{proposition}

A sufficient condition for finite expected lifespans is the Foster--Lyapunov
inequality \eqref{eq:dc}, with an increasing unbounded function $V$ in place
of output.

\begin{proposition}
    \label{l:life_drift}
    Let $V \colon \RR_+ \to \RR_+$ be increasing with $V(\phi) \to \infty$ as
    $\phi \to \infty$ and $\int V \, \diff \gamma < \infty$, and suppose there
    are $\lambda \in (0, 1)$ and $L < \infty$ such that
    \begin{equation}
        \label{eq:vdrift}
        \int V(\phi') \, \Gamma(\phi, \diff \phi')
        \leq \lambda V(\phi) + L
        \quad \text{for all } \phi \geq 0.
    \end{equation}
    Then $\EE \tau(p) < \infty$ for every $p > 0$ with $0 < \bar \phi(p) <
    \infty$.
\end{proposition}

Under the affine growth of Assumption~\ref{a:affine}(a) with the conditions of
Theorem~\ref{t:zipf}, $V(\phi) = \phi^s$ satisfies \eqref{eq:vdrift} for every
$0 < s < \alpha$, because $\EE A^s < 1$ for such $s$. Lemma~\ref{l:affdrift}
in the appendix gives $\lambda$ and $L$. Expected lifespan is therefore finite
for every $\alpha > 0$. At $\kappa \leq 1$, every equilibrium condition except
market clearing holds.

\subsection{Equilibrium under monopolistic competition}\label{ss:structures}

The firm's problem is that of Section~\ref{s:ee}, except that $p$ is now a
scalar the firm takes as given and the market-clearing condition
\eqref{eq:goods} has functions $h \colon \RR^2_+ \to \RR_+$ and $G \colon (0,
\infty) \to (0, \infty)$ in place of the output function $q$ and the demand
function $D$. Throughout this section, Assumptions~\ref{a:pq} to~\ref{a:gcm}
hold, other than Assumption~\ref{a:pgamma}(b). In particular, $\bar
v$, $\bar \phi$, $\tau$ and $p^*$ are defined as in \eqref{eq:vf},
\eqref{eq:barphi}, \eqref{eq:taustar} and \eqref{eq:eqenp}.

\begin{assumption}
    \label{a:g}
    The function $h \colon \RR^2_+ \to \RR_+$ is continuous. For each $p > 0$,
    $h(\cdot, p)$ is positive on $(0, \infty)$ and bounded, and $\inf_{\phi
    \geq a} h(\phi, p) > 0$ for every $a > 0$.
\end{assumption}

\begin{definition}
    \label{d:arr}
    A \emph{stationary equilibrium} for $(h, G)$ is a triple $(p, M, \mu)$ in
    $\eE$ satisfying the invariance condition \eqref{eq:invar}, the entry
    condition \eqref{eq:entc} and the balanced entry and exit condition
    \eqref{eq:beec}, and, in place of the goods-market condition
    \eqref{eq:goods}, the clearing condition
    \begin{equation}
        \label{eq:gclear}
        \int h(\phi, p) \, \mu(\diff \phi) = G(p).
    \end{equation}
\end{definition}

\begin{theorem}
    \label{t:arr}
    Let Assumption~\ref{a:g} hold. If $\bar \phi(p^*) < \infty$, then a unique
    stationary equilibrium for $(h, G)$ exists if and only if $\EE \tau(p^*) <
    \infty$.
\end{theorem}

Finite expected lifespan under affine growth (Proposition~\ref{l:life_drift})
gives existence for every tail index under the maintained
assumptions.

\begin{corollary}
    \label{p:arraff}
    Under the conditions of Theorem~\ref{t:arr} and of Theorem~\ref{t:zipf} on
    $(A, Y)$ and $\gamma$, a stationary equilibrium for $(h, G)$ exists for
    every $\alpha > 0$, and the productivity distribution $\mu^*$ has a Pareto
    tail with tail index $\alpha$.
\end{corollary}

Example~\ref{eg:jadditive} below develops a case with additive preferences
\citep{zhelobodko2012monopolistic,dhingra2019monopolistic} and a CES
subutility minus a linear term \citep{bulow1983note,mrazova2017not},
in which Assumptions~\ref{a:pq}(a) and~\ref{a:g} hold.\footnote{Demand
can also bound each firm's sales under the quadratic demand of
\cite{melitz2008market}, the CARA preferences of \cite{behrens2007general} and
Kimball demand \citep{kimball1995quantitative,klenow2016real,hubmer2026not}.}

\begin{example}
    \label{eg:jadditive}
    Assume a unit mass of identical consumers, each with utility $\int
    (x_i^\rho/\rho - k x_i) \, di$ over the industry's varieties, where $x_i$
    is the quantity of variety $i$, $0 < \rho < 1$ and $k > 0$. Each consumer
    spends exogenous income $y > 0$ on these varieties. Let $\lambda$ be the
    marginal utility of income where $p = 1/\lambda$ holds in equilibrium.
    Inverse demand for variety $i$ is $p_i = p(x_i^{\rho-1} - k)$, which firms
    take as given. A firm with productivity $\phi > 0$ produces $\phi n$ units
    with $n$ units of labor at a wage of one, so its marginal cost is
    $1/\phi$. Its profit-maximizing quantity is
    \begin{equation}
        \label{eq:jstatic}
        x(\phi,p)=\left(\frac{\rho p\phi}{1+pk\phi}\right)^{1/(1-\rho)},
    \end{equation}
    its sales are $\sS := p_i x = p(x^\rho - kx)$, and its profit is $\pi :=
    \sS - x/\phi - c_f$. Marginal revenue $p(\rho x^{\rho-1} - k)$ goes to
    zero at $\widehat x := (\rho/k)^{1/(1-\rho)}$, so output stays below
    $\widehat x$ and sales below $p(1-\rho)\widehat x^\rho$ at every
    productivity. Productivity follows Assumption~\ref{a:affine}(a) with $A >
    0$, and $(A, Y)$ and $\gamma$ satisfy Assumption~\ref{a:affine}(c) and the
    moment and nonarithmetic conditions of Theorem~\ref{t:zipf} for some
    $\alpha > 0$. The budget constraint $\int \sS(\phi,p)\,\mu(\diff\phi) = y$
    is \eqref{eq:gclear} with $h = \sS$ and $G \equiv y$, and sales are
    continuous, positive, increasing in $\phi$ and bounded at each $p$, so
    Assumption~\ref{a:g} holds. The entry condition fixes $p^*$, so $\lambda =
    1/p^*$ in equilibrium.\footnote{Profit is strictly increasing in $p$ and
    unbounded in $p$ at each $\phi > 0$, so Lemma~\ref{l:pphis} gives a unique
    $p^*$ at which the entry condition holds. The proof of Lemma~\ref{l:pphis}
    uses Assumptions~\ref{a:pq}, \ref{a:newent}(a) and~\ref{a:fc}, and not
    Assumption~\ref{a:pgamma}(b).} At $p^*$ the sales bound
    $p^*(1-\rho)\widehat x^\rho$ exceeds the fixed cost $c_f$, since otherwise
    no firm could ever profit and entry would not pay, and
    Lemma~\ref{l:jthreshold} in the appendix then gives $\bar\phi(p^*) <
    \infty$. Theorem~\ref{t:arr} and Corollary~\ref{p:arraff} therefore give a
    unique stationary equilibrium for every $\alpha > 0$.

\begin{table}[t]
\centering
\small
\begin{threeparttable}
\begin{tabular*}{\textwidth}{@{\extracolsep{\fill}} S[table-format=1.3] S[table-format=1.3] S[table-format=2.1] S[table-format=2.1] S[table-format=1.2] S[table-format=2.1] S[table-format=1.2] @{}}
\toprule
{$\alpha$} & {$p^*$} & {Expected lifespan} & {Mean size} & {Entry flow} & {Large-firm share} & {Sales index} \\
\midrule
1.500 & 2.126 & 11.7 & 3.0  & 2.62 & 30.2 & 1.51 \\
1.060 & 2.099 & 16.4 & 6.7  & 0.86 & 65.1 & 1.07 \\
1.030 & 2.097 & 16.9 & 7.2  & 0.77 & 67.8 & 1.04 \\
1.010 & 2.096 & 17.2 & 7.7  & 0.72 & 69.5 & 1.02 \\
1.000 & 2.095 & 17.4 & 7.9  & 0.69 & 70.4 & 1.01 \\
0.870 & 2.086 & 19.9 & 12.2 & 0.39 & 80.6 & 0.88 \\
\bottomrule
\end{tabular*}
\begin{tablenotes}[flushleft]
\footnotesize
\item \textit{Note}: Statistics as in Table~\ref{tab:boundary_main}.
Productivity process, entry distribution and costs as in
Figure~\ref{f:fds_pareto}, with the drift of $\ln A$ set to give the tail
index $\alpha$ of productivity in the first column. Demand as in
Example~\ref{eg:jadditive} with $\rho = 1/2$ and $k = 5 \times 10^{-5}$ (see
Appendix~\ref{s:comp} for further details).
\end{tablenotes}
\end{threeparttable}
\medskip
\caption{\label{tab:bounded}Monopolistic competition as $\alpha$ falls
  below one}
\end{table}

    Where $pk\phi$ is small, \eqref{eq:jstatic} gives $x$ approximately
    proportional to $\phi^{1/(1-\rho)}$ and sales $\sS \approx px^\rho$, so
    sales are approximately proportional to $\phi^{\rho/(1-\rho)}$. Where
    $pk\phi$ is large, sales approach the bound $p(1-\rho)\widehat x^\rho$.
    Sales therefore follow an approximate power law with index
    $\alpha(1-\rho)/\rho$ over the range where $pk\phi$ is small. Above the
    range, the elasticity of sales with respect to productivity falls toward
    zero as sales approach their bound, so an index fitted to larger firms
    will be higher.\footnote{\cite{kondo2023heavy} find the same
    pattern for employment in United States data, with the fitted index rising
    from 1.01 above 1,000 employees to 1.23 above 10,000.} The
    markup $p_i \phi = 1/\rho +
    (1-\rho)pk\phi/\rho$, however, grows in proportion to $\phi$ for every $k$
    and so retains the Pareto tail with index $\alpha$.

    Table~\ref{tab:bounded} illustrates equilibria with $\rho = 1/2$ in the
    economy of Figure~\ref{f:fds_pareto}, for tail indices of productivity
    from 1.5 down to 0.87. In every row the sales index is close
    to $\alpha$. Mean firm size no longer diverges, and the large-firm
    share rises gradually at and below Zipf's law, where no competitive
    equilibrium exists.

\end{example}

\section{Conclusion}

\label{s:c}

In this paper we investigated the entry-exit model of
\cite{hopenhayn1992entry} after removing the upper bound on firm productivity,
allowing us to consider more realistic representations of firm growth. In this
setting we provided an exact characterization of existence of stationary
equilibria as well as a Lyapunov-type sufficient condition for existence. We
also provided a new decomposition of the equilibrium distribution of firms, as
well as new representations of the rate of entry and aggregate supply.

Under affine productivity dynamics, which include Gibrat's law and
empirically relevant departures from  Gibrat's law,
the equilibrium firm size distribution
has a Pareto tail with an index set by the law of motion for incumbents and
the technology. A companion paper \citep{shanker2026pareto} extends this
result to a wide class of specifications of productivity growth. 

Finally, under price
taking, we showed market clearing rules out a
stationary equilibrium at and below Zipf's law. Under monopolistic competition
in which
demand bounds each firm's sales, affine dynamics give a stationary equilibrium
for every tail index of productivity, including those at or below one.



\appendix

\section{Proofs}\label{s:proofs}

In the proofs we use the operator notation
\begin{equation*}
    (\Gamma u)(\phi, p) := \int u(\phi', p) \Gamma(\phi, \diff \phi')
    \; \text{ for each } u \in \cf,
\end{equation*}
while $\pP$ denotes the Borel probability measures on $\RR_+$. The symbol
$\preceq$ represents first order stochastic dominance. All undefined notation
and terminology associated with Markov models follows \cite{meyn2012markov}.

Inserting the balanced entry and exit condition \eqref{eq:beec} into the
invariance condition \eqref{eq:invar} gives
\begin{equation}
    \label{eq:dlom}
    \mu_p (B)
    = \int \Pi_p(\phi, B) \mu_p (\diff \phi)
    \quad \text{for all } B \in \bB,
\end{equation}
where $\Pi_p$ is the transition kernel on $\RR_+$ defined by
\begin{equation}
    \label{eq:upm}
    \Pi_p(\phi, B)
    = \Gamma(\phi, B) \1\{\phi \geq \bar \phi(p) \} + \1\{\phi < \bar \phi(p)
    \} \gamma(B).
\end{equation}
Throughout the appendix, all assumptions in Section~\ref{s:ee} are in force.

\subsection{Preliminary Results}

{\sloppy This section contains proofs of preliminary results needed for the
main theorem. In the lemma below, we take $\delta \in (\beta, 1)$ from
Assumption~\ref{a:fc}.\par}

\begin{lemma}
    \label{l:gamprop}
    The operator $\Gamma$ is invariant on $\cf$ and $\Gamma \kappa \leq \kappa
    / \delta$.
\end{lemma}

\begin{proof}
    The last claim is easy to check, since, by the definition of $\kappa$ in
    \eqref{eq:dkappa}, we have
    \begin{equation}
        \label{eq:gk}
        \Gamma \kappa
        = \sum_{t \geq 0} \delta^t \Gamma^{t+1} \hat \pi
        = (1/\delta) \sum_{t \geq 0} \delta^{t+1} \Gamma^{t+1} \hat \pi
        \leq (1/\delta) \sum_{t \geq 0} \delta^t \Gamma^t \hat \pi
        = (1/\delta) \kappa.
    \end{equation}
    Now fix $u \in \cf$. That $\Gamma u$ is $\kappa$-bounded follows from the
    previous inequality and the pointwise bound $|u| \leq \| u \|_\kappa \,
    \kappa$. Continuity of $\Gamma u$ is immediate from
    Assumption~\ref{a:gcm}. Regarding monotonicity, let $u_n = u \1\{u \leq
    n\} + n \1\{u > n\}$ for each $n \in \NN$. Then $u_n$ is increasing for
    each $n$ and also bounded, so $\Gamma u_n$ is increasing for each
    $n$.\footnote{By Assumption~\ref{a:pgamma}, the kernel $\Gamma$ is
    monotone increasing. This implies that $\Gamma f$ is increasing whenever
    $f$ is measurable, increasing and bounded.} Moreover, by the Monotone
    Convergence Theorem, $\Gamma u_n \uparrow \Gamma u$. Since monotonicity is
    preserved under pointwise limits, $\Gamma u$ is also increasing. Hence
    $\Gamma u \in \cf$ as claimed.
\end{proof}

\begin{proof}[Proof of Lemma~\ref{l:tcm}]
    Pick any $u \in \cf$. Set $C := \max\{1,b-1\}$, with $b$ chosen as
    specified immediately below \eqref{eq:dkappa}. Using \eqref{eq:gk} and
    $|\pi| \leq C\kappa$,\footnote{By Assumption~\ref{a:pq}, $\pi(0,0)<0$, so
    the choice $\pi+b\geq1$ below \eqref{eq:dkappa} implies $b>1$. The $t=0$
    term of \eqref{eq:dkappa} gives $\kappa\geq\hat\pi=\pi+b\geq1$. We have
    $\pi\leq\hat\pi\leq\kappa\leq C\kappa$ and $-\pi\leq b-1\leq C\kappa$.
    Thus $|\pi|\leq C\kappa$.} we have
    \begin{equation*}
        |Tu | = \left| \pi + \beta \max \left\{ 0, \Gamma u \right\} \right|
        \leq |\pi| + \beta \Gamma |u |
        \leq C \kappa + \beta \| u \|_\kappa \Gamma \kappa
        \leq (C + \beta \| u \|_\kappa / \delta) \kappa.
    \end{equation*}
    Hence $\| Tu \|_\kappa$ is finite. In addition, $Tu$ is continuous and
    increasing because $Tu = \pi + \beta \max \left\{ 0, \Gamma u \right\}$
    and $\pi$ and $\Gamma u$ both have these properties (by
    Assumption~\ref{a:pq} and Lemma~\ref{l:gamprop}). Hence $T$ maps $\cf$
    into itself. In addition, $T$ is a contraction mapping, since, given $u,
    v$ in $\cf$,
    \begin{equation*}
        | Tu - Tv |
        \leq \beta \Gamma |u-v|
        \leq \beta \| u - v\|_\kappa \Gamma \kappa
        \leq (\beta/\delta) \| u - v\|_\kappa \kappa.
    \end{equation*}
    Dividing both sides by $\kappa$ and taking the supremum yields $d(Tu, Tv)
    \leq (\beta/\delta) d(u, v)$. Recalling that $\delta > \beta$, the claim
    of contractivity is established.

    To see that the fixed point $\bar v$ is strictly increasing in each
    argument on $(0, \infty)^2$, pick any $w \in \cf$ and observe that $Tw =
    \pi + \beta \max\{0, \Gamma w\}$ is strictly increasing in each argument
    on $(0, \infty)^2$, since $\Gamma w$ is increasing by
    Lemma~\ref{l:gamprop} and $\pi$ is strictly increasing in each argument on
    $(0, \infty)^2$ by Assumption~\ref{a:pq}. In other words, $T$ maps
    elements of $\cf$ into functions strictly increasing on $(0, \infty)^2$.
    Given that $\bar v = T\bar v$, the function $\bar v$ must itself be
    strictly increasing in each argument on $(0, \infty)^2$.

    Finally, we prove $\bar v(\phi, 0) = \pi(\phi, 0)$ for all $\phi \geq 0$.
    By Assumption~\ref{a:pq}, the function $\pi$ is continuous and increasing,
    and by the footnote above, $|\pi|\leq C\kappa$. Hence $\pi\in\cf$. Also by
    Assumption~\ref{a:pq}, $\pi(\phi', 0) < 0$ for every $\phi' \geq 0$. Since
    $\Gamma(\phi,\cdot)$ is a probability measure, $(\Gamma\pi)(\phi, 0) < 0$
    for every $\phi \geq 0$. Hence $(T\pi)(\phi, 0) = \pi(\phi, 0)$. If
    $u\in\cf$ satisfies $u(\cdot,0)=\pi(\cdot,0)$, then $(\Gamma
    u)(\phi,0)=(\Gamma\pi)(\phi,0)<0$ and $(Tu)(\phi,0)=\pi(\phi,0)$. Thus
    $(T^n\pi)(\phi, 0) = \pi(\phi, 0)$ for every $n \geq 1$. Since $T$ is a
    contraction on $(\cf,d)$, $d(T^n\pi,\bar v)\to0$. For each $\phi\geq0$,
    the definition of $d$ and $\kappa(\phi,0)<\infty$ give
    \begin{align*}
        |\bar v(\phi,0)-\pi(\phi,0)|
        &= |\bar v(\phi,0)-(T^n\pi)(\phi,0)| \\
        &\leq \kappa(\phi,0)d(\bar v,T^n\pi)\to0.
    \end{align*}
    Hence $\bar v(\phi,0)=\pi(\phi,0)<0$ by Assumption~\ref{a:pq}.
\end{proof}

\begin{proof}[Proof of Lemma~\ref{l:barphi}]
    Let $\Phi(p) := \{ \phi \geq 0 \; | \; (\Gamma \bar v)(\phi, p) \geq 0
    \}$. This set is nonempty when $p > 0$ by $\bar v \geq \pi$ and
    Assumption~\ref{a:pgamma}. Moreover, $\Phi(p)$ is closed because, if $\{
    \phi_n \} \subset \Phi(p)$ and $\phi_n \to \phi$, then, by the continuity
    in Assumption~\ref{a:gcm}, we have $0 \leq (\Gamma \bar v)(\phi_n, p) \to
    (\Gamma \bar v)(\phi, p)$. Hence $(\Gamma \bar v)(\phi, p) \geq 0$ and,
    therefore, $\phi \in \Phi(p)$. Since $\Phi(p)$ is closed and nonempty when
    $p > 0$, $\bar \phi(p) = \min \Phi (p)$ exists in $\RR_+$.

    Due to monotonicity of $\bar v$, the correspondence $\Phi$ is such that $p
    \leq p'$ implies $\Phi(p) \subset \Phi(p')$. Hence $\bar \phi(p) = \min
    \Phi (p)$ is decreasing.

    At $p=0$, $\bar v(\phi,0) = \pi(\phi,0) < 0$ for every $\phi \geq 0$ by
    Lemma~\ref{l:tcm}. Since $\Gamma$ is a probability kernel, $(\Gamma\bar
    v)(\phi,0) < 0$ for every $\phi \geq 0$. Hence $\Phi(0) = \varnothing$ and
    $\bar\phi(0) = \infty$. For $p > 0$, $\Gamma\bar v \in \cf$ by
    Lemma~\ref{l:gamprop}, so $\phi \mapsto (\Gamma\bar v)(\phi,p)$ is
    increasing. Since $\bar\phi(p)=\min\Phi(p)$, monotonicity gives $\Phi(p) =
    [\bar\phi(p), \infty)$. Hence $(\Gamma\bar v)(\phi,p) \geq 0$ if and only
    if $\phi \geq \bar\phi(p)$ for $p > 0$ and $\phi \geq 0$.
\end{proof}

\begin{proof}[Proof of Lemma~\ref{l:pphis}]
    Let $e(p) := \int \bar v( \phi, p) \gamma(\diff \phi) - c_e$. The function
    $e$ is finite on $\RR_+$ because $|\bar v| \leq K\kappa$, where $K :=
    \|\bar v\|_\kappa < \infty$ by Lemma~\ref{l:tcm}, and
    \begin{equation*}
        \int \kappa( \phi, p) \gamma(\diff \phi)
        = \sum_{t \geq 0} \delta^t
        \EE \, \hat \pi( \phi_t, p) < \infty
    \end{equation*}
    by Assumption~\ref{a:fc}. The function $e$ is also continuous on $\RR_+$.
    To see this, take $p_n \to p$. Since convergent sequences are bounded, we
    can choose $\tilde p$ such that $p_n \leq \tilde p$ for all $n$. By
    monotonicity, it follows that $\bar v(\phi', p_n) \leq \bar v(\phi',
    \tilde p)$ for all $\phi'$. Since $\bar v \geq 1-b$, continuity of $\bar
    v$ and the Dominated Convergence Theorem now give $e(p_n) \to e(p)$.

    If $p = 0$, then, by Lemma~\ref{l:tcm}, we have $\bar v(\phi, p)<0$ for
    all $\phi$, so $e(p) < 0$. Conversely, if $p$ is large enough, then $\int
    \pi( \phi, p) \gamma(\diff \phi) \geq c_e$ by
    Assumption~\ref{a:newent}(a). As $\bar v \geq \pi$, this implies that
    $e(p) \geq 0$. Hence, by the Intermediate Value Theorem, there exists a $p
    > 0$ such that $e(p) = 0$.

    Assumptions~\ref{a:pq} and~\ref{a:newent}(a) give $\gamma((0,\infty))>0$.
    By Lemma~\ref{l:tcm}, $\bar v$ is increasing in price on $\RR_+^2$ and
    strictly increasing in price on $(0,\infty)^2$. Hence $e$ is strictly
    increasing on $(0,\infty)$, and has at most one zero. The price $p > 0$
    with $e(p) = 0$ is therefore unique.
\end{proof}

\begin{proof}[Proof of Corollary~\ref{l:pintv}]
    Fix $\bar p$ satisfying Assumption~\ref{a:newent}. We first show
    $(\Gamma\bar v)(0,\bar p)<0$. Suppose instead that $(\Gamma\bar v)(0,\bar
    p)\geq0$. By Lemma~\ref{l:tcm}, $\bar v\in\cf$, and $\Gamma$ is monotone
    increasing by Assumption~\ref{a:pgamma}, so $\Gamma\bar v$ is increasing
    by Lemma~\ref{l:gamprop}. Hence $(\Gamma\bar v)(\phi,\bar
    p)\geq(\Gamma\bar v)(0,\bar p)\geq0$ for every $\phi\geq0$. Therefore
    \eqref{eq:vf} gives $\bar v(\phi,\bar p) = \pi(\phi,\bar p) +
    \beta(\Gamma\bar v)(\phi,\bar p)$ for every $\phi\geq0$. Applying $\Gamma$
    to the last equation and iterating forward gives, for $n\geq1$,
    \begin{equation}\label{eq:sum_iter}
        (\Gamma\bar v)(0,\bar p)
        = \sum_{t=1}^{n}\beta^{t-1}(\Gamma^t\pi)(0,\bar p)
        + \beta^n(\Gamma^{n+1}\bar v)(0,\bar p).
    \end{equation}
    By the proof of Lemma~\ref{l:tcm} and $\bar v\in\cf$, there exists
    $K<\infty$ such that $|\pi|+|\bar v|\leq K\kappa$. Lemma~\ref{l:gamprop}
    gives $\Gamma\kappa\leq\kappa/\delta$, and iterating this bound gives
    $\Gamma^t\kappa\leq\delta^{-t}\kappa$ for $t\geq1$. Since $\kappa(0,\bar
    p)<\infty$, as noted below \eqref{eq:dkappa}, and $\beta/\delta<1$, we
    have
    \begin{equation*}
        \sum_{t\geq1}\beta^{t-1}(\Gamma^t|\pi|)(0,\bar p)
        \leq \frac{K\kappa(0,\bar p)}{\delta-\beta}<\infty.
    \end{equation*}
    Hence the sum in \eqref{eq:sum_iter} converges as $n\to\infty$, and
    \begin{equation*}
        \left|\beta^n(\Gamma^{n+1}\bar v)(0,\bar p)\right|
        \leq \frac{K\kappa(0,\bar p)}{\delta}
            \left(\frac{\beta}{\delta}\right)^n\longrightarrow0.
    \end{equation*}
    Letting $n\to\infty$ in \eqref{eq:sum_iter}, multiplying by $\beta$ and
    applying Assumption~\ref{a:newent}(b) gives
    \begin{equation*}
        0 \leq \beta(\Gamma\bar v)(0,\bar p)
        = \sum_{t\geq1}\beta^t(\Gamma^t\pi)(0,\bar p)<0,
    \end{equation*}
    where the first inequality holds because $\beta>0$ and $(\Gamma\bar
    v)(0,\bar p)\geq0$. No number satisfies both bounds, so $(\Gamma\bar
    v)(0,\bar p)<0$.

    We now show the exit threshold is strictly positive at every price in
    $(0,\bar p]$. Fix $p\in(0,\bar p]$. By Lemma~\ref{l:tcm}, $\bar v$ is
    increasing in price, thus $(\Gamma\bar v)(0,p)\leq(\Gamma\bar v)(0,\bar
    p)<0$. In turn, Lemma~\ref{l:barphi} gives $\bar\phi(p)>0$.

    It remains to show $\bar\phi(p^*)>0$. Let $e(p) := \int \bar v(\phi, p) \,
    \gamma(\diff \phi) - c_e$, as in the proof of Lemma~\ref{l:pphis}. The
    Bellman equation \eqref{eq:vf} gives $\bar v\geq\pi$, so the entry
    condition of Assumption~\ref{a:newent}(a) gives
    \begin{equation*}
        e(\bar p)
        \geq \int \pi(\phi, \bar p) \, \gamma(\diff \phi) - c_e
        \geq 0.
    \end{equation*}
    By Lemma~\ref{l:pphis} and the definition of $p^*$ in \eqref{eq:eqenp}, we
    have $p^* > 0$ and $e(p^*) = 0$. The proof of Lemma~\ref{l:pphis} shows
    $e$ is strictly increasing on $(0,\infty)$. If $p^* > \bar p$, then $0 =
    e(p^*) > e(\bar p) \geq 0$, which is impossible. Hence $p^* \leq \bar p$.
    Since $0 < p^* \leq \bar p$, the previous paragraph gives $\bar\phi(p^*) >
    0$.
\end{proof}

\begin{proof}[Proof of the claim in Remark~\ref{r:converse}]
    Fix $p > 0$. By \eqref{eq:vf}, $\bar v \geq \pi + \beta \Gamma \bar v$.
    Since $\Gamma$ is a positive linear operator, applying $\Gamma$ and
    iterating gives, for $n \geq 1$,
    \begin{equation*}
        (\Gamma\bar v)(0,p)
        \geq \sum_{t=1}^{n}\beta^{t-1}(\Gamma^t\pi)(0,p)
        + \beta^n(\Gamma^{n+1}\bar v)(0,p).
    \end{equation*}
    The bounds in the proof of Corollary~\ref{l:pintv} hold at any price,
    since $\kappa(0, p) < \infty$, so the sum converges and the last term
    converges to zero as $n \to \infty$. Multiplying by $\beta$ gives $\beta (\Gamma
    \bar v)(0, p) \geq S(p)$. If $S(p) \geq 0$, then $(\Gamma \bar v)(0, p)
    \geq 0$, and Lemma~\ref{l:barphi} gives $\bar \phi(p) = 0$. Conversely, if
    $\bar \phi(p) = 0$, then $(\Gamma \bar v)(0, p) \geq 0$ by
    Lemma~\ref{l:barphi}, and the first paragraph of the proof of
    Corollary~\ref{l:pintv}, with $p$ in place of $\bar p$, gives $\beta
    (\Gamma \bar v)(0, p) = S(p)$. Hence $S(p) \geq 0$. Thus $\bar \phi(p) >
    0$ if and only if $S(p) < 0$.
\end{proof}

\begin{lemma}
    \label{l:qpirr}
    For all $p > 0$ with $\bar \phi(p)>0$, the transition kernel $\Pi_p$
    defined in \eqref{eq:upm} is aperiodic, $\gamma$-irreducible and admits
    the accessible atom $\mathbf a_p := [0, \bar \phi(p))$.
\end{lemma}

\begin{proof}
    Fix $\phi \in \RR_+$ and let $B$ be any Borel set such that $\gamma(B) >
    0$. Let $\{\phi_t\}$ be a Markov process on $\RR_+$ generated by $\Pi_p$
    and starting at $\phi_0 = \phi$. Evidently, if $\phi < \bar \phi(p)$, then
    $\Pi_p(\phi, B) = \gamma(B) > 0$, so $B$ is reachable from $\phi$. If
    instead $\phi \geq \bar \phi(p)$, then we let $m$ be the smallest $n \in
    \NN$ such that $\Gamma^n(\phi, [0, \bar \phi(p))) > 0$. By the
    Chapman--Kolmogorov equations, we have
    \begin{equation*}
        \PP\{ \phi_{m+1} \in B\}
         = \int \Pi_p(\phi', B) \Pi_p^m(\phi, \diff \phi')
         \geq \int_0^{\bar \phi(p)} \Pi_p(\phi', B) \Pi_p^m(\phi, \diff
         \phi').
    \end{equation*}
    The right-hand side is at least $ \gamma(B) \, \Gamma^m(\phi, [0, \bar
    \phi(p)))$. Here $\Gamma^m$ replaces $\Pi_p^m$ because $m$ is minimal, so
    the chain has not entered $[0, \bar \phi(p))$ before date $m$ and the two
    kernels agree up to that date. The bound is strictly positive by the
    assumed positivity of $\gamma(B)$ and the definition of $m$. Again $B$ is
    reachable, and hence $\Pi_p$ is $\gamma$-irreducible. Aperiodicity now
    follows from Assumption~\ref{a:aper}.

    Finally, $\gamma(\mathbf a_p) > 0$ by Assumption~\ref{a:aper}. The
    interval $\mathbf a_p$ is an atom because $\Pi_p(\phi, A) = \Pi_p(\psi, A)
    = \gamma(A)$ for all $\phi, \psi < \bar \phi(p)$.
\end{proof}

The next lemma discusses petite sets, as defined in \cite{meyn2012markov}.

\begin{lemma}
    \label{l:small}
    If $p > 0$ satisfies $\bar \phi(p)>0$ and $d > 0$, then $[0, d]$ is a
    petite set for $\Pi_p$.
\end{lemma}

\begin{proof}
    Fix $p, d > 0$. It suffices to show existence of a nontrivial Borel
    measure $\nu$ on $\RR_+$ and an $m \in \NN$ such that $\Pi_p^m(\phi, B)
    \geq \nu(B)$ whenever $0 \leq \phi \leq d$ and $B \in \bB$. Let $\mathbf
    a_p := [0, \bar \phi(p))$ and take the smallest $n \in \NN$ such that
    $\epsilon := \Gamma^n(d, \mathbf a_p) > 0$. (This $n$ is finite by
    Assumption~\ref{a:pgamma}.) Assumption~\ref{a:aper} gives $\eta :=
    \gamma(\mathbf a_p) > 0$. Pick any $\phi \in [0, d]$ and $B \in \bB$, and
    let $\{ \phi_t \}$ be generated by $\Pi_p$ from initial condition $\phi$.
    Let $H := \inf\{t \geq 0 : \phi_t \in \mathbf a_p\}$. Until date $H$, the
    chain uses $\Gamma$, so $H$ has the same law as the first-hit time of the
    $\Gamma$-chain started at $\phi$. A visit to $\mathbf a_p$ at date $n$
    requires a first visit by that date. By monotonicity of $\Gamma$,
    \begin{equation*}
        \PP\{H \leq n\}
        \geq \Gamma^n(\phi, \mathbf a_p)
        \geq \Gamma^n(d, \mathbf a_p) = \epsilon.
    \end{equation*}
    After a visit to $\mathbf a_p$, each transition remains there with
    probability $\eta$. Hence
    \begin{equation*}
        \PP\{\phi_n \in \mathbf a_p\}
        \geq \sum_{k=0}^n \PP\{H=k\}\eta^{n-k}
        \geq \eta^n \PP\{H \leq n\}
        \geq \epsilon\eta^n.
    \end{equation*}
    By the law of total probability,
    \begin{equation*}
        \Pi_p^{n+1}(\phi, B)
        = \PP \{ \phi_{n+1} \in B \}
        \geq \PP \{ \phi_{n+1} \in B \,|\, \phi_n \in \mathbf a_p  \}
        \PP \{ \phi_n \in \mathbf a_p \}.
    \end{equation*}
    By the definition of $\Pi_p$ and the bound $\PP\{\phi_n\in\mathbf
    a_p\}\geq\epsilon\eta^n$, we then have
    \begin{equation*}
        \Pi_p^{n+1}(\phi, B)
        \geq \gamma(B) \PP \{ \phi_n \in \mathbf a_p\}
        \geq \gamma(B)\epsilon\eta^n.
    \end{equation*}
    Setting $\nu := \epsilon\eta^n\gamma$ and $m := n+1$ therefore gives
    $\Pi_p^m(\phi, B) \geq \nu(B)$, which verifies the claim in the lemma.
\end{proof}

\begin{lemma}
    \label{l:uni_stat}
    Fix $p > 0$ with $\bar \phi(p)>0$ and let $\Pi_p$ and $\mathbf a_p$ be as
    in Lemma~\ref{l:qpirr}. The following statements are equivalent:
    \begin{enumerate}
        \item[\rm{(i)}] There exists a $\mu_p \in \pP$ such that
            \eqref{eq:dlom} holds.
        \item[\rm{(ii)}] Expected firm lifespan $\EE \tau(p)$ is finite.
    \end{enumerate}
    If either and hence both of these conditions holds, then $\mu_p$ is
    unique, $\mu_p (\mathbf a_p) > 0$ and
    \begin{equation}\label{eq:pof}
        \mu_p (B) =  \mu_p (\mathbf a_p)
        \cdot \EE \sum_{t = 0}^{\tau(p)} \1\{\phi_t \in B\}
    \end{equation}
    for all $B \in \bB$. In addition, aggregate supply obeys
    \begin{equation}\label{eq:asob}
        0 < \int q(\phi, p) \mu_p (\diff \phi)
            =  \mu_p (\mathbf a_p) \EE \, \ell(p).
    \end{equation}
\end{lemma}
\begin{proof}[Proof of Lemma~\ref{l:uni_stat}]
    Fix $p > 0$ with $\bar \phi(p) > 0$. The kernel $\Pi_p$ is
    $\gamma$-irreducible by Lemma~\ref{l:qpirr}. The set $\mathbf a_p = [0,
    \bar \phi(p))$ is an atom for $\Pi_p$ because $\Pi_p(\phi, B) = \gamma(B)$
    whenever $\phi \in \mathbf a_p$. Moreover, $\gamma(\mathbf a_p) > 0$ by
    Assumption~\ref{a:aper}. By \eqref{eq:upm}, a $\Pi_p$-chain started in
    $\mathbf a_p$ draws its next state from $\gamma$ and then follows $\Gamma$
    until it returns to $\mathbf a_p$. The states of the chain from date $1$
    to the return date therefore have the same law as
    $\phi_0,\ldots,\phi_{\tau(p)}$ in \eqref{eq:inff}. Hence the return time
    has the same law as $\tau(p)+1$. It now follows from Theorem~10.2.2 of
    \cite{meyn2012markov} that $\Pi_p$ is positive recurrent -- and hence
    admits a stationary probability $\mu_p$ -- if and only if the expected
    return time to $\mathbf a_p$ is finite. This is equivalent to $\EE \tau(p)
    < \infty$, which proves the first claim in the lemma.

    Now suppose that $\EE \tau(p) < \infty$ holds and let $\mu_p$ be
    stationary for $\Pi_p$. Then $\Pi_p$ is positive recurrent, and hence
    recurrent by Proposition~10.1.1 of \cite{meyn2012markov}, so
    Theorem~10.4.9 of the same reference gives \eqref{eq:pof} and makes the
    invariant measure unique up to a constant multiple. Thus $\mu_p$ is the
    only stationary probability for $\Pi_p$. Positivity of $\mu_p (\mathbf
    a_p)$ holds because $\mathbf a_p$ is an accessible atom
    (Lemma~\ref{l:qpirr}). Since output $q$ is nonnegative, \eqref{eq:pof}
    extends via the Monotone Convergence Theorem to
    \begin{equation*}
        \int q(\phi, p) \mu_p (\diff \phi)
        =  \mu_p (\mathbf a_{p})
            \cdot \EE \sum_{t = 0}^{\tau(p)} q(\phi_t, p)
            =  \mu_p (\mathbf a_{p})
            \EE \, \ell(p),
    \end{equation*}
    which gives the equality in \eqref{eq:asob}. The entry condition in
    Assumption~\ref{a:newent}(a) implies $\gamma((0,\infty))>0$, as shown in
    the proof of Lemma~\ref{l:pphis}. Positivity of supply follows from
    $q(\phi, p) > 0$ for all $\phi > 0$ and $\mu_p (\mathbf a_p) > 0$, since
    $\ell(p)\geq q(\phi_0,p)$ and
    \begin{equation*}
        \int q(\phi, p) \mu_p (\diff \phi)
        \geq \mu_p(\mathbf a_p)\int q(\phi,p)\gamma(\diff\phi) > 0.
        \qedhere
    \end{equation*}
\end{proof}

\subsection{Existence and Uniqueness}

In what follows, $p^*$ is the equilibrium entry price, as defined in
\eqref{eq:eqenp}. Corollary~\ref{l:pintv} gives $\bar \phi(p^*) > 0$.

\begin{proof}[Proof of Theorem~\ref{t:bk1}]
    Suppose first that $\EE \, \ell(p^*)=\infty$ and yet there exists a pair
    $(M^*, \mu^*)$ such that $(p^*, M^*, \mu^*)$ is a stationary equilibrium.
    By \eqref{eq:beec}, $\mu^*(\mathbf a^*) = M^*$, and market clearing
    \eqref{eq:goods} with $q(\bar \phi(p^*),p^*) > 0$ gives
    \begin{equation*}
        \mu^*([\bar \phi(p^*), \infty)) \leq D(p^*)/q(\bar \phi(p^*),p^*),
    \end{equation*}
    so $m := \mu^*(\RR_+)$ lies in $(0,\infty)$ and $\mu_{p^*} := \mu^*/m$ is
    a stationary probability for $\Pi_{p^*}$. This means that (i) in
    Lemma~\ref{l:uni_stat} holds. Hence we can apply \eqref{eq:asob},
    obtaining
    \begin{equation*}
        \int q(\phi, p^*) \mu^* (\diff \phi)
        = \mu^*(\mathbf a^*)\EE \, \ell(p^*) = \infty.
    \end{equation*}
    Since $\mu^*(\mathbf a^*) > 0$, aggregate supply is infinite, while demand
    $D(p^*)$ is finite. Contradiction.

    Now suppose instead that $\EE \, \ell(p^*)$ is finite. Observe that, by
    monotonicity of $q$ and the definition of $\tau(p^*)$, we have
    %
    \begin{equation*}
        \ell(p^*)
        = \sum_{t=0}^{\tau(p^*)} q(\phi_t, p^*)
        \geq \sum_{t=0}^{\tau(p^*)-1} q(\bar \phi(p^*), p^*)
        = \tau(p^*) q(\bar \phi(p^*),p^*).
    \end{equation*}
    %
    \begin{equation*}
        \fore
        \tau(p^*) \leq c \sum_{t=0}^{\tau(p^*)} q(\phi_t, p^*)
        \text{ for some } c > 0.
    \end{equation*}
    Taking expectations gives $\EE \tau(p^*) \leq c \EE \, \ell(p^*) <
    \infty$. Hence, by Lemma~\ref{l:uni_stat}, the distribution $\mu_{p^*}$
    satisfying \eqref{eq:dlom} at $p^*$ is well defined, and \eqref{eq:pof}
    and \eqref{eq:asob} both hold at $p^*$. Let
    \begin{equation}\label{eq:mustar}
        s := \frac{D(p^*)}{\int q(\phi, p^*) \mu_{p^*}(\diff \phi)},
        \qquad
        \mu^* := s \, \mu_{p^*}
        \quad \text{and} \quad
        M^* \coloneq \mu^*\setntn{\phi}{\phi < \bar \phi(p^*)},
    \end{equation}
    where $s$ is well defined and positive by \eqref{eq:asob}.

    Set $\phi^* := \bar \phi(p^*)$ and $\mu := \mu_{p^*}$. By construction,
    the triple $(p^*, M^*, \mu^*)$ satisfies all of the conditions in
    Section~\ref{ss:sredef}. For example, the goods market clears because, by
    the definition of $s$ in \eqref{eq:mustar}, we have $\int q( \phi, p^*)
    \mu^*(\diff \phi) = s \int q( \phi, p^*) \mu(\diff \phi) = D(p^*)$. The
    time-invariance condition~\ref{eq:invar} holds because, given any $B$ in
    $\bB$,
    \begin{align*}
        \int \Gamma(\phi, B) \1\{\phi \geq \bar \phi(p^*)\} \mu^*(\diff \phi)
        & + M^* \gamma(B)
        \\
         & = \int \Gamma(\phi, B) \1\{\phi \geq \phi^* \} \mu^*(\diff \phi)
            + \mu^* \{\phi < \phi^*\} \gamma(B)
            \\
         & = s
             \int
             \left[
                 \Gamma(\phi, B) \1\{\phi \geq \phi^* \}
                 + \1 \{\phi < \phi^*\} \gamma(B)
             \right]
             \mu(\diff \phi).
    \end{align*}
    Since $\mu$ satisfies \eqref{eq:dlom}, this last expression is just $s \mu
    (B)$, or, equivalently, $\mu^*(B)$. Hence $(p^*, M^*, \mu^*)$ is a
    stationary equilibrium.

    The triple $(p^*, M^*, \mu^*)$ is the only stationary equilibrium in $\eE$
    because the stationary probability $\mu_{p^*}$ is unique, by
    Lemma~\ref{l:uni_stat}, and the price $p^*$ is uniquely determined by
    Lemma~\ref{l:pphis}. As in the first paragraph, any equilibrium measure at
    $p^*$ has finite mass, so it is $m \mu_{p^*}$ for some $m > 0$, and market
    clearing \eqref{eq:goods} determines $m$. Given $p^*$ and $\mu^*$, the
    constant $M^*$ is then uniquely determined by \eqref{eq:beec}.

    The decomposition \eqref{eq:kdec} follows from \eqref{eq:pof}. Multiplying
    both sides by $s$ from \eqref{eq:mustar} gives
    \begin{equation*}
        \mu^* (B)
        =  \mu^* (\mathbf a^*)
            \cdot \EE \sum_{t = 0}^{\tau(p^*)} \1\{\phi_t \in B\}
        =  M^*
            \cdot \EE \sum_{t = 0}^{\tau(p^*)} \1\{\phi_t \in B\}.
    \end{equation*}
    Equation \eqref{eq:ase} follows from \eqref{eq:asob} and the definition of
    $M^*$.
\end{proof}

\subsection{Drift}

\begin{proof}[Proof of Proposition~\ref{p:geoerg}]
    Adopting the conditions of Proposition~\ref{p:geoerg}, we first show that
    $\Pi_{p}$ is $V$-uniformly ergodic for all $p > 0$ with $\bar \phi(p) > 0$
    via Theorem~16.1.2 of \cite{meyn2012markov}. Lemma~\ref{l:qpirr} shows
    that $\Pi_p$ is irreducible and aperiodic, so we need only show that the
    drift condition (V4) defined in Chapter~15 of \cite{meyn2012markov} holds
    with respect to a petite set. By Lemma~15.2.8 of the same reference, this
    will be true whenever there exists a finite function $V \geq 1$ on $\RR_+$
    such that the sublevel set $C_a := \setntn{\phi \in \RR_+}{V(\phi) \leq
    a}$ is petite for each $a \geq 0$ and, for some positive constants $\alpha
    < 1$ and $K < \infty$,
    \begin{equation}
        \label{eq:mtdc}
        \int V(\phi') \Pi_p(\phi, \diff \phi')
            \leq \alpha V(\phi) + K
            \text{ for all } \phi \geq 0.
    \end{equation}
    Set $V(\phi) = 1+q(\phi,p)$. Since $q$ is real valued and nonnegative by
    Assumption~\ref{a:pq}, we have $1 \leq V < \infty$, as \eqref{eq:mtdc}
    requires. By Assumption~\ref{a:unbq}, $q(\phi,p)\to\infty$ as
    $\phi\to\infty$. Hence every $C_a$ is contained in $[0,d]$ for some $d>0$.
    Since $[0,d]$ is petite by Lemma~\ref{l:small}, so is $C_a$. Moreover, for
    any fixed $\phi \geq 0$, the definition of $\Pi_p$ and the
    Foster--Lyapunov condition \eqref{eq:dc} yield
    \begin{equation*}
        \int V(\phi') \Pi_p(\phi, \diff \phi')
        \leq \int V(\phi') \gamma(\diff \phi')
            + \int V(\phi') \Gamma(\phi, \diff \phi') .
    \end{equation*}
    The first term is finite by Assumption~\ref{a:aper}. The second term is
    bounded by $\lambda V(\phi)+1-\lambda+L$ using \eqref{eq:dc}. Increase $L$
    to be nonnegative if necessary. Putting these bounds together gives
    \eqref{eq:mtdc} with $\alpha := \lambda$ and $K := \int
    V(\phi')\gamma(\diff\phi')+1-\lambda+L$.

    Next we claim that $\EE \, \ell(p)$ is finite for this same arbitrary $p$.
    To see this, let $\mu$ be the unique stationary distribution of $\Pi_p$,
    existence of which is guaranteed by $V$-uniform ergodicity (see
    Theorem~16.1.2 of \cite{meyn2012markov}). The term $m(p):=\int
    q(\phi,p)\mu(\diff\phi)$ is finite by Proposition~4.24 of
    \cite{hairer2006ergodic} and \eqref{eq:mtdc}, since
    \begin{equation*}
        m(p) \leq \int V(\phi)\mu(\diff\phi)
        \leq \frac{K}{1-\lambda} < \infty.
    \end{equation*}
    Lemma~\ref{l:uni_stat} gives $m(p)=\mu(\mathbf a_p)\EE\ell(p)$ and
    $\mu(\mathbf a_p)>0$. Hence $\EE\ell(p)<\infty$.

    As $\EE \, \ell(p)$ is finite for all positive $p$ with $\bar\phi(p)>0$,
    it is finite at $p^*$ by Corollary~\ref{l:pintv}. By Theorem~\ref{t:bk1},
    the model has a unique stationary equilibrium.
\end{proof}

\begin{proof}[Proof of Proposition~\ref{p:affnec}]
    Let $\rho := \EE A^\eta \in [0, \infty]$. By Theorem~\ref{t:bk1},
    statement (a) holds if and only if $\EE \, \ell(p^*) < \infty$. We show
    that (c) implies (b), (b) implies (a), and (a) implies (c).

    \emph{(c) implies (b).} Fix $p > 0$. It suffices to find $\lambda \in (0,
    1)$ and $L < \infty$ with $\EE (A \phi + Y)^\eta \leq \lambda \phi^\eta +
    L$ for all $\phi \geq 0$, since multiplying by $m(p)$ then gives
    \eqref{eq:dc}. If $\eta \leq 1$, then $(a + b)^\eta \leq a^\eta + b^\eta$
    for $a, b \geq 0$, so $\EE (A \phi + Y)^\eta \leq \rho \phi^\eta + \EE
    Y^\eta$, and we can take $\lambda := \max\{\rho, 1/2\}$ and $L := \EE
    Y^\eta$. If $\eta > 1$, then Minkowski's inequality in $L^\eta$ gives $\EE
    (A \phi + Y)^\eta \leq (\rho^{1/\eta} \phi + k)^\eta$ with $k := (\EE
    Y^\eta)^{1/\eta}$. Fix $\lambda \in (\rho, 1)$. Since $(\rho^{1/\eta} \phi
    + k)^\eta / \phi^\eta \to \rho < \lambda$ as $\phi \to \infty$, there is a
    $\phi_1$ with $(\rho^{1/\eta} \phi + k)^\eta \leq \lambda \phi^\eta$ for
    $\phi \geq \phi_1$, and we can take $L := (\rho^{1/\eta} \phi_1 +
    k)^\eta$.

    \emph{(b) implies (a).} Under Assumption~\ref{a:affine}, $q(\phi, p) =
    m(p) \phi^\eta \to \infty$ as $\phi \to \infty$ for each $p > 0$, so
    Assumption~\ref{a:unbq} holds, and (a) follows from
    Proposition~\ref{p:geoerg}.

    \emph{(a) implies (c).} Since (a) is equivalent to $\EE \, \ell(p^*) <
    \infty$, we show that $\EE \, \ell(p^*) = \infty$ whenever (c) fails. Let
    $c := \bar \phi(p^*)$, which is positive by Corollary~\ref{l:pintv}, and
    write $m := m(p^*)$. Let $\EE_x$ denote expectation when the process in
    \eqref{eq:inff} starts at $\phi_0 = x$. Since $\phi_0 \sim \gamma$, we
    have $\EE \, \ell(p^*) = \int \EE_x \ell(p^*) \, \gamma(\diff x)$, and
    $\gamma([c, \infty)) > 0$ by Assumption~\ref{a:affine}(c). Hence it
    suffices to show that $\EE_x \ell(p^*) = \infty$ for every $x \geq c$
    whenever $\EE Y^\eta = \infty$ or $\rho \geq 1$. Fix $x \geq c$. Since
    $\phi_0 = x \geq c$, we have $\tau(p^*) \geq 1$, and so $\ell(p^*) \geq
    q(\phi_1, p^*) = m \phi_1^\eta$.

    Suppose first that $\EE Y^\eta = \infty$. Then $\EE_x \ell(p^*) \geq m \EE
    (A x + Y)^\eta \geq m \EE Y^\eta = \infty$.

    Suppose next that $\rho \geq 1$. Let $\psi_0 := x$ and $\psi_t := x A_1
    \cdots A_t$ for $t \geq 1$, and let $\sigma := \inf\setntn{t \geq
    0}{\psi_t < c}$. Since each $Y_t$ is nonnegative, induction on $t$ gives
    $\phi_t \geq \psi_t$ for all $t$. Hence $\sigma \leq \tau(p^*)$ and
    \begin{equation}\label{eq:affcomp}
        \ell(p^*)
        \geq m \sum_{t=0}^{\tau(p^*)} \phi_t^\eta
        \geq m \sum_{t=0}^{\sigma - 1} \psi_t^\eta.
    \end{equation}
    If $\rho = \infty$, then $\EE_x \ell(p^*) \geq m \EE_x \phi_1^\eta \geq m
    x^\eta \rho = \infty$. If $A$ is almost surely constant, then $\rho \geq
    1$ forces $A \geq 1$, so $\psi_t \geq x \geq c$ for all $t$, $\sigma =
    \infty$, and \eqref{eq:affcomp} gives $\ell(p^*) = \infty$.

    It remains to treat $1 \leq \rho < \infty$ with $A$ nondegenerate. Let
    $\tilde \PP$ be a probability under which $A_1, A_2, \ldots$ are {\sc iid}
    with $\tilde \PP\{A \in \diff a\} = a^\eta \PP\{A \in \diff a\} / \rho$.
    For each $t \geq 0$, the event $\{\sigma > t\}$ is determined by $A_1,
    \ldots, A_t$, and $\psi_t^\eta = x^\eta \prod_{i=1}^t A_i^\eta$, so
    \begin{equation*}
        \EE_x \left[ \psi_t^\eta \1\{\sigma > t\} \right]
        = x^\eta \rho^t \, \tilde \PP\{\sigma > t\}
        \geq x^\eta \, \tilde \PP\{\sigma > t\}.
    \end{equation*}
    Summing over $t$ and using \eqref{eq:affcomp} gives $\EE_x \ell(p^*) \geq
    m x^\eta \sum_{t \geq 0} \tilde \PP\{\sigma > t\} = m x^\eta \, \tilde \EE
    \sigma$. We now show that $\tilde \PP\{\sigma = \infty\} > 0$, which
    implies $\tilde \EE \sigma = \infty$.

    Under $\tilde \PP$, $A > 0$ almost surely, so $S_t := \sum_{i=1}^t \ln
    A_i$ is finite, and $\sigma = \inf\setntn{t \geq 0}{S_t < \ln(c/x)}$. The
    negative part of $\ln A$ is $\tilde \PP$-integrable, since $\tilde \EE
    (\ln A)^- = \EE[A^\eta (\ln A)^-] / \rho \leq 1/(e \eta \rho)$. Let $B :=
    A^\eta$, so that $\EE B = \rho$ and $\EE [B \ln B] = \eta \rho \, \tilde
    \EE \ln A$. The function $b \mapsto b \ln b$ is strictly convex and
    bounded below on $\RR_+$, and $B$ is nondegenerate, so Jensen's inequality
    gives $\EE[B \ln B] > \rho \ln \rho \geq 0$. Hence $\tilde \EE \ln A \in
    (0, \infty]$, and $S_t \to \infty$ $\tilde \PP$-almost surely by the
    strong law of large numbers.

    We claim that $\tilde \PP\{S_t > 0 \text{ for all } t \geq 1\} > 0$. If
    not, then $\sigma_1 := \inf\setntn{t \geq 1}{S_t \leq 0}$ is finite
    $\tilde \PP$-almost surely. Since the increments of $S$ are {\sc iid}, the
    strong Markov property shows that the successive times $\sigma_{k+1} :=
    \inf\setntn{t > \sigma_k}{S_t \leq S_{\sigma_k}}$ are also finite almost
    surely. Then $S_{\sigma_k} \leq 0$ for all $k$, which contradicts $S_t \to
    \infty$. This proves the claim. Since $\ln(c/x) \leq 0$, the event in the
    claim is contained in $\{\sigma = \infty\}$, so $\tilde \PP\{\sigma =
    \infty\} > 0$, as required.
\end{proof}

\subsection{Pareto Tails}

Throughout this subsection, Assumption~\ref{a:affine} and the conditions of
Theorem~\ref{t:zipf} are in force. Let $g(s) := \EE A^s$ for $s \geq 0$, with
$A^0 := \1\{A > 0\}$.

\begin{lemma}
    \label{l:convex}
    For $s > 0$, we have $g(s) < 1$ if and only if $s < \alpha$.
\end{lemma}

\begin{proof}
    The function $g$ is finite and convex on $[0, \alpha + \delta]$, with
    $g(0) = \PP\{A > 0\} \leq 1$ and $g(\alpha) = 1$. Since the distribution
    of $\ln A$ given $A > 0$ is nonarithmetic, $A$ is not almost surely
    constant on $\{A > 0\}$, so $g$ is strictly convex on $[0, \alpha +
    \delta]$. If $0 < s < \alpha$, then strict convexity gives $g(s) < (1 -
    s/\alpha) g(0) + (s/\alpha) g(\alpha) \leq 1$. If $s > \alpha$ and $g(s)
    \leq 1$, then $g(s)$ is finite and strict convexity gives $g(\alpha) <
    \max\{g(0), g(s)\} \leq 1$, a contradiction. Hence $g(s) > 1$ for $s >
    \alpha$.
\end{proof}

In the next two lemmas we assume $\alpha > \eta$, so that, by
Lemma~\ref{l:convex}, Proposition~\ref{p:affnec} and Theorem~\ref{t:bk1}, $\EE
\, \ell(p^*) < \infty$ and the stationary probability $\mu := \mu_{p^*}$ of
Lemma~\ref{l:uni_stat} exists. Let $c := \bar \phi(p^*)$, which is positive by
Corollary~\ref{l:pintv}, and let $\mathbf a := [0, c)$.

\begin{lemma}
    \label{l:moments}
    If $\alpha > \eta$, then $\int \phi^s \mu(\diff \phi) < \infty$ for every
    $s \in (0, \alpha)$.
\end{lemma}

\begin{proof}
    Fix $s \in (0, \alpha)$. By Lemma~\ref{l:convex}, $\EE A^s < 1$. In
    addition, $\EE Y^s \leq 1 + \EE Y^\alpha < \infty$ and $\int z^s
    \gamma(\diff z) < \infty$. The argument for (c) implies (b) in the proof
    of Proposition~\ref{p:affnec}, with $s$ in place of $\eta$, gives $\lambda
    \in (0, 1)$ and $L \geq 0$ with $\EE (A \phi + Y)^s \leq \lambda \phi^s +
    L$ for all $\phi \geq 0$. Let $V(\phi) := 1 + \phi^s$. As in the proof of
    Proposition~\ref{p:geoerg},
    \begin{equation*}
        \int V(\phi') \Pi_{p^*}(\phi, \diff \phi')
        \leq \lambda V(\phi) + K
        \quad \text{with} \quad
        K := \int V \diff \gamma + 1 - \lambda + L.
    \end{equation*}
    Iterating gives $\Pi_{p^*}^n V \leq \lambda^n V + K/(1 - \lambda)$ for all
    $n$. Fix $N \in \NN$. Since $\mu$ is stationary,
    \begin{equation*}
        \int (V \wedge N) \diff \mu
        = \int \Pi_{p^*}^n (V \wedge N) \diff \mu
        \leq \int \left[ \left( \lambda^n V + \frac{K}{1-\lambda} \right)
            \wedge N \right] \diff \mu.
    \end{equation*}
    Letting $n \to \infty$ and using dominated convergence gives $\int (V
    \wedge N) \diff \mu \leq K / (1 - \lambda)$. Letting $N \to \infty$ and
    using monotone convergence gives $\int V \diff \mu \leq K / (1 - \lambda)
    < \infty$.
\end{proof}

\begin{lemma}
    \label{l:lower}
    If $\alpha > \eta$, then there exist $c_0 > 0$ and $x_0 < \infty$ such
    that $\mu(x, \infty) \geq c_0 x^{-\alpha}$ for all $x \geq x_0$.
\end{lemma}

\begin{proof}
    Let $\{\phi_t\}$ be as in \eqref{eq:inff} and $\tau := \tau(p^*)$. By
    \eqref{eq:pof},
    \begin{equation}\label{eq:kaclow}
        \mu(x, \infty)
        = \mu(\mathbf a) \, \EE \sum_{t = 0}^{\tau} \1\{\phi_t > x\}
        = \mu(\mathbf a) \int \EE_y \sum_{t = 0}^{\tau} \1\{\phi_t > x\}
            \, \gamma(\diff y),
    \end{equation}
    where $\EE_y$ denotes expectation when $\phi_0 = y$, and $\mu(\mathbf a) >
    0$ by Lemma~\ref{l:uni_stat}. By Assumption~\ref{a:affine}(c), we can fix
    $y_1 > c$ with $\gamma([c, y_1]) > 0$.

    Fix $y \in [c, y_1]$ and $x > y_1$. With $\phi_0 = y$, let $\psi_t$ and
    $\sigma$ be as in the proof of Proposition~\ref{p:affnec}, so that $\phi_t
    \geq \psi_t = y A_1 \cdots A_t$ and $\sigma \leq \tau$. Then
    \begin{equation}\label{eq:lowpsi}
        \EE_y \sum_{t = 0}^{\tau} \1\{\phi_t > x\}
        \geq \PP\{\psi_t > x \text{ for some } t < \sigma\}.
    \end{equation}
    Let $\tilde \PP$ be a probability under which $A_1, A_2, \ldots$ are {\sc
    iid} with $\tilde \PP\{A \in \diff a\} = a^\alpha \PP\{A \in \diff a\}$.
    Under $\tilde \PP$, $A > 0$ almost surely, and $\xi_i := \ln A_i$ has mean
    $\tilde m := \EE[A^\alpha \ln A]$ and satisfies $\tilde \EE (\xi_i^+)^2 =
    \EE[A^\alpha (\ln^+ A)^2] < \infty$, by $\EE A^{\alpha + \delta} <
    \infty$. As in the proof of Proposition~\ref{p:affnec}, Jensen's
    inequality gives $\tilde m > 0$, and $\tilde m$ is finite. Let $S_t :=
    \xi_1 + \cdots + \xi_t$, $u := \ln(x/y) > 0$, $T := \inf\setntn{t \geq
    1}{S_t > u}$ and $O := S_T - u$. Let $E := \{T < \infty \text{ and } S_t
    \geq 0 \text{ for } t \leq T\}$. On $E$, we have $\psi_t = y e^{S_t} \geq
    y \geq c$ for $t \leq T$, so $\sigma > T$, and $\psi_T > x$. Hence the
    right-hand side of \eqref{eq:lowpsi} is at least $\PP(E)$. The event $E
    \cap \{T = t\}$ is determined by $A_1, \ldots, A_t$ and is contained in
    $\{A_1 \cdots A_t > 0\}$, so
    \begin{equation*}
        \PP(E)
        = \sum_{t \geq 1} \tilde \EE \left[ e^{-\alpha S_t}
            \1_{E \cap \{T = t\}} \right]
        = \tilde \EE \left[ e^{-\alpha S_T} \1_E \right]
        = \left( \frac{y}{x} \right)^\alpha
            \tilde \EE \left[ e^{-\alpha O} \1_E \right].
    \end{equation*}
    Under $\tilde \PP$, $S_t \to \infty$ almost surely by the strong law of
    large numbers, so $T < \infty$ almost surely, and $p_0 := \tilde \PP\{S_t
    > 0 \text{ for all } t \geq 1\}$ is positive by the argument in the proof
    of Proposition~\ref{p:affnec}. By Lorden's inequality
    \citep[Theorem~1]{lorden1970excess}, $\tilde \EE O \leq L_0 := \tilde \EE
    (\xi_1^+)^2 / \tilde m$, where $L_0$ does not depend on $u$. Let $k := 2
    L_0 / p_0$, so that $\tilde \PP\{O > k\} \leq p_0/2$ by Markov's
    inequality. Since $\{S_t > 0 \text{ for all } t \geq 1\} \cap \{T <
    \infty\} \subset E$,
    \begin{equation*}
        \tilde \EE \left[ e^{-\alpha O} \1_E \right]
        \geq e^{-\alpha k} \,
            \tilde \PP \left( \{S_t > 0 \text{ for all } t \geq 1\}
            \cap \{O \leq k\} \right)
        \geq e^{-\alpha k} \, \frac{p_0}{2}.
    \end{equation*}
    Combining these bounds with $y \geq c$ gives $\EE_y \sum_{t=0}^{\tau}
    \1\{\phi_t > x\} \geq c^\alpha e^{-\alpha k} (p_0/2) x^{-\alpha}$ for all
    $y \in [c, y_1]$ and $x > y_1$. Inserting this into \eqref{eq:kaclow}
    proves the lemma with $x_0 := 2 y_1$ and $c_0 := \mu(\mathbf a) \gamma([c,
    y_1]) c^\alpha e^{-\alpha k} p_0 / 2$.
\end{proof}

\begin{proof}[Proof of Theorem~\ref{t:zipf}]
    By Proposition~\ref{p:affnec}, the model has a unique stationary
    equilibrium $(p^*, M^*, \mu^*)$ if and only if $\EE A^\eta < 1$ and $\EE
    Y^\eta < \infty$. Note that $\kappa > 1$ if and only if $\alpha > \eta$.
    If $\alpha > \eta$, then $\EE A^\eta < 1$ by Lemma~\ref{l:convex}, and
    $\EE Y^\eta \leq 1 + \EE Y^\alpha < \infty$, so an equilibrium exists. If
    $\alpha \leq \eta$, then $\EE A^\eta \geq 1$ by Lemma~\ref{l:convex}, so
    no equilibrium exists. This proves (b) and the existence claim in (a).

    It remains to prove the tail claim in (a). Suppose $\alpha > \eta$, and
    let $\mu$, $c$ and $\mathbf a$ be as above. Let $X$, $(A, Y)$ and $Z$ be
    independent, with $X \sim \mu$ and $Z \sim \gamma$, and let
    \begin{equation*}
        R := (A X + Y) \1\{X \geq c\} + Z \1\{X < c\}.
    \end{equation*}
    By \eqref{eq:upm}, $R$ has distribution $\int \Pi_{p^*}(\phi, \cdot)
    \mu(\diff \phi)$, which equals $\mu$ by stationarity. We first show that
    $\EE |R^\alpha - (A X)^\alpha| < \infty$. On $\{X < c\}$, we have
    $|R^\alpha - (AX)^\alpha| \leq Z^\alpha + c^\alpha A^\alpha$, which is
    integrable. On $\{X \geq c\}$, if $\alpha \leq 1$, then $|R^\alpha -
    (AX)^\alpha| \leq Y^\alpha$, which is integrable. If $\alpha > 1$, then
    the mean value theorem and the bound $(a + b)^{r} \leq 2^{(r-1)^+} (a^r +
    b^r)$ for $a, b, r \geq 0$ give
    \begin{equation*}
        |R^\alpha - (AX)^\alpha|
        \leq \alpha Y (AX + Y)^{\alpha - 1}
        \leq \alpha \, 2^{(\alpha - 2)^+}
            \left( Y A^{\alpha - 1} X^{\alpha - 1} + Y^\alpha \right).
    \end{equation*}
    By independence and H\"older's inequality,
    \begin{equation*}
        \EE[Y A^{\alpha - 1} X^{\alpha - 1}]
        \leq (\EE Y^\alpha)^{1/\alpha} (\EE A^\alpha)^{(\alpha - 1)/\alpha}
            \int \phi^{\alpha - 1} \mu(\diff \phi),
    \end{equation*}
    which is finite by Lemma~\ref{l:moments}. Hence $\EE |R^\alpha -
    (AX)^\alpha| < \infty$.

    Next, for nonnegative random variables $U$ and $V$, Fubini's theorem gives
    \begin{equation*}
        \int_0^\infty |\PP\{U > t\} - \PP\{V > t\}| \, \alpha t^{\alpha -
        1} \diff t
        \leq \EE \int_0^\infty |\1\{U > t\} - \1\{V > t\}| \, \alpha
        t^{\alpha - 1} \diff t
        = \EE |U^\alpha - V^\alpha|.
    \end{equation*}
    Applying this with $U = R$ and $V = AX$, and using $\PP\{R > t\} = \PP\{X
    > t\}$, we obtain $\int_0^\infty |\PP\{X > t\} - \PP\{AX > t\}| \,
    t^{\alpha - 1} \diff t < \infty$. Since $X$ and $A$ are independent, $\EE
    A^\alpha = 1$, $\EE A^\alpha \ln^+ A < \infty$ and $\ln A$ given $A > 0$
    is nonarithmetic, Theorem~2.3 of \cite{goldie1991implicit} now gives
    $\lim_{x \to \infty} x^\alpha \mu(x, \infty) = C$ for some $C \in [0,
    \infty)$. By Lemma~\ref{l:lower}, $C \geq c_0 > 0$.

    Finally, by the proof of Theorem~\ref{t:bk1}, $\mu^* = s \mu$ for some $s
    > 0$. With $m := m(p^*)$,
    \begin{equation*}
        x^{\alpha/\eta} \, \nu^*(x, \infty)
        = s \, x^{\alpha/\eta} \, \mu \left( (x/m)^{1/\eta}, \infty
            \right)
        \to s \, C \, m^{\alpha/\eta}
        \quad \text{as } x \to \infty,
    \end{equation*}
    and the limit lies in $(0, \infty)$. Hence $\nu^*$ has a Pareto tail with
    tail index $\alpha/\eta = \kappa$.
\end{proof}

\subsection{Below Zipf's Law}\label{ss:jproofs}

\begin{lemma}
    \label{l:affdrift}
    Let $(A, Y)$ be a nonnegative random vector and $s > 0$. If $\EE A^s < 1$
    and $\EE Y^s < \infty$, then there exist $\lambda \in (0, 1)$ and $L <
    \infty$ such that $\EE (A \phi + Y)^s \leq \lambda \phi^s + L$ for all
    $\phi \geq 0$.
\end{lemma}

\begin{proof}
    Let $\rho := \EE A^s$. If $s \leq 1$, then $(a + b)^s \leq a^s + b^s$ for
    $a, b \geq 0$, so $\EE (A \phi + Y)^s \leq \rho \phi^s + \EE Y^s$. Take
    $\lambda := \max\{\rho, 1/2\}$ and $L := \EE Y^s$. If $s > 1$, then
    Minkowski's inequality in $L^s$ gives $\EE (A \phi + Y)^s \leq (\rho^{1/s}
    \phi + k)^s$ with $k := (\EE Y^s)^{1/s}$. Fix $\lambda \in (\rho, 1)$.
    Since $(\rho^{1/s} \phi + k)^s / \phi^s \to \rho < \lambda$ as $\phi \to
    \infty$, there is a $\phi_1$ with $(\rho^{1/s} \phi + k)^s \leq \lambda
    \phi^s$ for $\phi \geq \phi_1$. Take $L := (\rho^{1/s} \phi_1 + k)^s$. For
    $\phi < \phi_1$ the bound $(\rho^{1/s} \phi + k)^s$ is increasing in
    $\phi$ and so at most $L$. Hence $\EE (A \phi + Y)^s \leq \lambda \phi^s +
    L$ for all $\phi \geq 0$.
\end{proof}

\begin{proof}[Proof of Proposition~\ref{p:demog}]
    We have $p^* > 0$ by \eqref{eq:eqenp} and $\bar\phi(p^*) > 0$ by
    Corollary~\ref{l:pintv}. Hence Lemma~\ref{l:uni_stat} applies at $p = p^*$
    with $\mathbf a_{p^*} = [0, \bar\phi(p^*))$. For (a), suppose first that
    $(p^*, M, \mu)$ satisfies \eqref{eq:invar} and \eqref{eq:beec} for some
    finite measure $\mu$ on $\bB$. Since $\mu(\mathbf a_{p^*}) = M > 0$ by
    \eqref{eq:beec} and $\mu$ is finite, $m := \mu(\RR_+)$ lies in $(0,
    \infty)$. Inserting \eqref{eq:beec} into \eqref{eq:invar} gives
    \eqref{eq:dlom} for $\mu$ at $p^*$, so the probability measure $\mu / m$
    is stationary for $\Pi_{p^*}$. Lemma~\ref{l:uni_stat} therefore gives $\EE
    \tau(p^*) < \infty$.

    Conversely, suppose $\EE \tau(p^*) < \infty$. Lemma~\ref{l:uni_stat} gives
    a unique stationary probability $\mu_{p^*}$ for $\Pi_{p^*}$, and
    \eqref{eq:pof} with $B = \RR_+$ gives $\mu_{p^*}(\mathbf a_{p^*}) \,
    \EE[\tau(p^*) + 1] = 1$, so $\mu_{p^*} = \Lambda_{p^*} / \EE[\tau(p^*) +
    1]$. The finite measure $\mu_M := M \Lambda_{p^*} = M \EE[\tau(p^*) + 1]
    \, \mu_{p^*}$ satisfies $\mu_M(\mathbf a_{p^*}) = M$, which is
    \eqref{eq:beec}, and it is stationary for $\Pi_{p^*}$, so \eqref{eq:dlom}
    with $\mu_M(\mathbf a_{p^*}) = M$ is \eqref{eq:invar}. This proves (a).
    For (b), suppose $\EE \tau(p^*) < \infty$ and let $\mu$ be a finite
    measure such that $(p^*, M, \mu)$ satisfies \eqref{eq:invar} and
    \eqref{eq:beec}. As in the proof of (a), $\mu / \mu(\RR_+) = \mu_{p^*}$,
    so $M = \mu(\mathbf a_{p^*}) = \mu(\RR_+) / \EE[\tau(p^*) + 1]$ and $\mu =
    \mu_M$. The entry condition \eqref{eq:entc} holds at $p^*$ by the
    definition \eqref{eq:eqenp}. This proves (b).
\end{proof}

\begin{proof}[Proof of Proposition~\ref{l:life_drift}]
    Fix $p > 0$ with $0 < \bar \phi(p) < \infty$. Let $W := 1 + V$. Then $1
    \leq W < \infty$. Since $V(\phi) \to \infty$ as $\phi \to \infty$, each
    sublevel set of $W$ is contained in some $[0, d]$ and so is petite by
    Lemma~\ref{l:small}. Condition \eqref{eq:vdrift} with $\int V \, \diff
    \gamma < \infty$ gives \eqref{eq:mtdc} for $W$ with $\alpha := \lambda$
    and $K := \int W \, \diff \gamma + 1 - \lambda + L$. The first part of the
    proof of Proposition~\ref{p:geoerg} therefore applies with $W$ in place of
    $1 + q(\cdot, p)$. Hence $\Pi_p$ is $W$-uniformly ergodic and has a
    stationary distribution, and Lemma~\ref{l:uni_stat} gives $\EE \tau(p) <
    \infty$.
\end{proof}

\begin{proof}[Proof of Theorem~\ref{t:arr}]
    By Corollary~\ref{l:pintv} and the condition of the theorem, $0 <
    \bar\phi(p^*) < \infty$. Let $\underline h := \inf_{\phi \geq
    \bar\phi(p^*)} h(\phi, p^*)$, $\bar h := \sup_{\phi \geq 0} h(\phi, p^*)$
    and $\ell_h := \sum_{t=0}^{\tau(p^*)} h(\phi_t, p^*)$. Since $0 <
    \bar\phi(p^*) < \infty$, Assumption~\ref{a:g} gives $0 < \underline h \leq
    \bar h < \infty$. Since $h \geq 0$, the Monotone Convergence Theorem
    applied to the definition of $\Lambda_{p^*}$ in Proposition~\ref{p:demog}
    gives
    \begin{equation}
        \label{eq:ellh}
        \int h(\phi, p^*) \, \Lambda_{p^*}(\diff\phi)
        = \EE \sum_{t=0}^{\tau(p^*)} h(\phi_t, p^*)
        = \EE \ell_h.
    \end{equation}

    Suppose first that $(p, M, \mu)$ is a stationary equilibrium for $(h, G)$.
    The entry condition \eqref{eq:entc} gives $p = p^*$ by
    Lemma~\ref{l:pphis}. Since $\mu([0, \bar\phi(p^*))) = M$ by
    \eqref{eq:beec} and $h(\cdot, p^*) \geq \underline h$ on $[\bar\phi(p^*),
    \infty)$, the clearing condition \eqref{eq:gclear} gives $\mu(\RR_+) \leq
    M + G(p^*) / \underline h < \infty$. Proposition~\ref{p:demog}(a)
    therefore gives $\EE \tau(p^*) < \infty$, and part (b) gives $\mu = M
    \Lambda_{p^*}$. Substituting $M \Lambda_{p^*}$ for $\mu$ in
    \eqref{eq:gclear} and using \eqref{eq:ellh} gives $M \EE \ell_h = G(p^*)$,
    so $\EE \ell_h < \infty$ and $M = G(p^*) / \EE \ell_h$. Hence a stationary
    equilibrium requires $\EE \tau(p^*) < \infty$, and at most one exists.

    Conversely, suppose $\EE \tau(p^*) < \infty$. Since $h(\cdot, p^*) \leq
    \bar h$, $\ell_h \leq (\tau(p^*) + 1) \, \bar h$, so $\EE \ell_h \leq \bar
    h \, \EE[\tau(p^*) + 1] < \infty$. As $\ell_h \geq h(\phi_0, p^*)$ and
    $\phi_0 \sim \gamma$, $\EE \ell_h \geq \int h(\phi, p^*) \,
    \gamma(\diff\phi)$, and this integral is positive because $h(\cdot, p^*) >
    0$ on $(0, \infty)$ and $\gamma((0, \infty)) > 0$ by
    Assumptions~\ref{a:pq} and~\ref{a:newent}(a). Hence $M^* := G(p^*) / \EE
    \ell_h$ lies in $(0, \infty)$. Set $\mu^* := M^* \Lambda_{p^*}$.
    Proposition~\ref{p:demog} gives \eqref{eq:invar}, \eqref{eq:entc} and
    \eqref{eq:beec}, and \eqref{eq:ellh} gives $\int h(\phi, p^*) \,
    \mu^*(\diff\phi) = M^* \EE \ell_h = G(p^*)$, which is \eqref{eq:gclear}.
    Hence $(p^*, M^*, \mu^*)$ is a stationary equilibrium. The equilibrium is
    unique by the first part of the proof.
\end{proof}

\begin{proof}[Proof of Corollary~\ref{p:arraff}]
    First, $p^* > 0$ by \eqref{eq:eqenp}, and $0 < \bar\phi(p^*) < \infty$ by
    Corollary~\ref{l:pintv} and the condition of Theorem~\ref{t:arr}. Fix $s
    \in (0, \alpha)$. By Lemma~\ref{l:convex}, $\EE A^s < 1$. Both $\EE Y^s
    \leq 1 + \EE Y^\alpha$ and $\int \phi^s \gamma(\diff\phi) \leq 1 + \int
    \phi^\alpha \gamma(\diff\phi)$ are finite. Hence Lemma~\ref{l:affdrift}
    and Assumption~\ref{a:affine}(a) give \eqref{eq:vdrift} for the increasing
    unbounded function $V(\phi) := \phi^s$, and Proposition~\ref{l:life_drift}
    gives $\EE \tau(p^*) < \infty$. Theorem~\ref{t:arr} gives a unique
    stationary equilibrium $(p^*, M^*, \mu^*)$. 
    
    Since $\mu^*$ is finite by
    \eqref{eq:beec}, \eqref{eq:gclear} and Assumption~\ref{a:g},
    Proposition~\ref{p:demog}(b) and \eqref{eq:pof} give $\mu^* = M^*
    \Lambda_{p^*} = M^* \EE[\tau(p^*) + 1] \, \mu_{p^*}$, where $\mu_{p^*}$ is
    the stationary probability of Lemma~\ref{l:uni_stat}. In the proofs of
    Lemmas~\ref{l:moments} and~\ref{l:lower}, and in the step of the proof
    of Theorem~\ref{t:zipf} that establishes $x^\alpha \mu_{p^*}(x, \infty)
    \to C > 0$, the condition $\alpha > \eta$ enters only through $\EE
    \tau(p^*) < \infty$, which holds here by Proposition~\ref{l:life_drift}.
    Apart from this condition, the proofs of Lemmas~\ref{l:moments}
    and~\ref{l:lower} and of Theorem~\ref{t:zipf} rely only on
    Assumption~\ref{a:affine}(a) and~(c), the moment and nonarithmetic
    conditions of Theorem~\ref{t:zipf}, Lemmas~\ref{l:convex}
    and~\ref{l:affdrift}, Lemma~\ref{l:uni_stat} at $p^*$ and the kernel
    \eqref{eq:upm} with $0 < \bar\phi(p^*) < \infty$, all of which hold
    here. The proofs therefore apply at $p^*$ and show that $x^\alpha
    \mu_{p^*}(x, \infty)$ converges to a limit in $(0, \infty)$, so $\mu^*$
    has a Pareto tail with tail index $\alpha$.
\end{proof}

\begin{lemma}
    \label{l:jthreshold}
    Under the assumptions of Example~\ref{eg:jadditive}, $0 < \bar\phi(p^*) <
    \infty$.
\end{lemma}

\begin{proof}
    Corollary~\ref{l:pintv} gives $\bar\phi(p^*) > 0$. It remains to show that
    $\bar\phi(p^*) < \infty$. Let $\widehat x := (\rho/k)^{1/(1-\rho)}$ and
    $\bar\pi := p^*(1-\rho)\widehat x^\rho - c_f$. By \eqref{eq:jstatic},
    $x(\phi, p^*) \to \widehat x$ as $\phi \to \infty$, so $\pi(\phi, p^*) =
    \sS(\phi, p^*) - x(\phi, p^*)/\phi - c_f \to \bar\pi$. Since $\pi(\cdot,
    p^*)$ is increasing by Assumption~\ref{a:pq}, if $\bar\pi \leq 0$ then
    $\pi(\cdot, p^*) < 0$. The last paragraph of the proof of
    Lemma~\ref{l:tcm}, with $p^*$ in place of $0$, then gives $\bar v(\cdot,
    p^*) = \pi(\cdot, p^*) < 0$, which contradicts $\int \bar v(\phi, p^*)
    \gamma(\diff\phi) = c_e > 0$. Hence $\bar\pi > 0$. Choose $a > 0$ such
    that $\pi(\phi', p^*) \geq \bar\pi/2$ whenever $\phi' \geq a$. Since $\bar
    v \geq \pi \geq -c_f$, $(\Gamma\bar v)(\phi, p^*) \geq \bar\pi/2 -
    (\bar\pi/2 + c_f) \, \Gamma(\phi, [0, a))$. Since $Y \geq 0$ and $A > 0$,
    $\Gamma(\phi, [0, a)) \leq \PP\{A \leq a/\phi\} \to \PP\{A \leq 0\} = 0$
    as $\phi \to \infty$, so $(\Gamma\bar v)(\phi, p^*) > 0$ for large $\phi$.
    Since $\Gamma\bar v$ is continuous by Assumption~\ref{a:gcm}, the set in
    \eqref{eq:barphi} at $p^*$ is closed and nonempty, so $\bar\phi(p^*) <
    \infty$.
\end{proof}

\section{Computing the Equilibrium}\label{s:comp}

Under the finite expected lifetime output condition of
Theorem~\ref{t:bk1}, the unique stationary equilibrium can be
computed as follows:
\begin{enumerate}
    \item[(S1)] Obtain $\bar v$ as the unique fixed point of $T$ in $\cf$ and
        $\bar \phi$ as in \eqref{eq:barphi}.
    \item[(S2)] Solve for the equilibrium entry price $p^*$, as in
        Lemma~\ref{l:pphis}.
    \item[(S3)] Define $\Pi_{p^*}$ via \eqref{eq:upm} and compute $\mu_{p^*}$
        as the unique solution to \eqref{eq:dlom} in $\pP$.
    \item[(S4)] Rescale $\mu_{p^*}$ by setting
        $s := D(p^*)/ \int q( \phi, p^*) \mu_{p^*}(\diff \phi)$ and then
        $\mu^* := s \, \mu_{p^*}$.
    \item[(S5)] Obtain the mass of entrants via
        $M^* = \mu^* \{ \phi < \bar \phi(p^*) \}$.
\end{enumerate}

The proof of Theorem~\ref{t:bk1} confirms that the triple $(p^*, M^*,
\mu^*)$ computed via (S1)--(S5) is a stationary equilibrium.

Regarding (S1), $\bar v$ is a fixed point of a contraction map, as shown in
Lemma~\ref{l:tcm}. This provides the basis of a globally convergent method of
computation. The value $p^*$ in (S2) can be obtained once we solve for $\bar v$
and $\bar \phi$.

Uniqueness in (S3) holds because $\Pi_p$ in \eqref{eq:upm} is
$\gamma$-irreducible (see \cite{meyn2012markov} for definitions) whenever
$\bar \phi(p) > 0$, and in particular at $p^*$, by Corollary~\ref{l:pintv}. The
condition $\EE \tau(p^*) < \infty$ then implies that the same process is
positive recurrent, and it is aperiodic by Lemma~\ref{l:qpirr}. A chain started
from $\gamma$ reaches the atom $[0, \bar \phi(p^*))$ with probability one, since
its hitting time has the law of $\tau(p^*)$. Standard results for chains with
an accessible atom \citep[Chapters~13 and~17]{meyn2012markov} then give
convergence of its marginal distributions to $\mu_{p^*}$ and of its sample path
averages to integrals under $\mu_{p^*}$, opening avenues for computing
$\mu_{p^*}$ through either simulation or successive
approximations.\footnote{Stronger statements are true when
    Assumption~\ref{a:unbq} holds and productivity growth obeys the
    Foster--Lyapunov condition. We show in the proof of
    Proposition~\ref{p:geoerg} that in this case the transition kernel $\Pi_p$
    is $V$-uniformly ergodic \cite[Chapter~16]{meyn2012markov} for all $p$ with
    $\bar \phi(p) > 0$, implying that the marginal distributions generated by
    $\Pi_{p}$ converge to its unique stationary distribution at a geometric rate
    and yielding a range of sample path properties.}

Rescaling in (S4) is implemented so that the goods market clears.

One nontrivial issue with the computation in (S3) is that, as
Theorem~\ref{t:zipf} shows, the productivity distribution $\mu_{p^*}$ and hence
the firm size distribution have Pareto tails under realistic firm-level growth
dynamics. This complicates numerics. Methods for handling fat tails numerically
have been proposed by \cite{gouin2019pareto} in the context of the wealth
distribution and similar ideas should be applicable here.

\begin{figure}[t]
	\centering
	\scalebox{0.6}{\includegraphics[trim=11 11 10 10, clip]{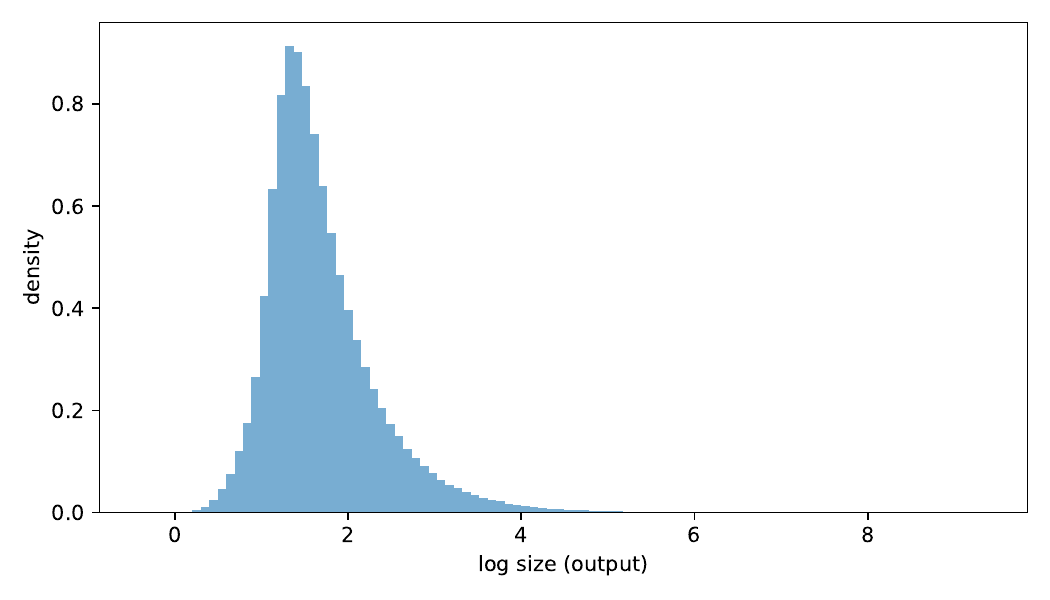}}
	\caption{\label{f:fds_hist}Histogram of the simulated log firm size
    distribution}
	\vspace{-\baselineskip}
\end{figure}

The equilibrium shown in Figure~\ref{f:fds_pareto} and Figure~\ref{f:fds_hist}
is computed in this way for Example~\ref{eg:glcd} with $\beta = 0.95$, $\theta
= 0.3$, $w = 1$, $c_f = 4$, $c_e = 1$, demand $D(p) = 1/p$, $\ln A$ normal
with mean $-0.012$ and standard deviation $0.1$, and entrants with $\ln \phi$
normal with mean $1$ and standard deviation $0.2$. Step (S1) uses value
function iteration on a grid of $2{,}000$ points on $[0, 50]$, with integrals
over $A$ and $\gamma$ computed by Monte Carlo, and step (S2) uses bisection.
Step (S3) is implemented by simulating one million firms under $\Pi_{p^*}$ for
$4{,}000$ periods, with no upper bound on productivity.
Figure~\ref{f:fds_hist} shows a histogram of log firm size. The distribution
looks roughly lognormal in its center, but its right tail is much heavier, as
Figure~\ref{f:fds_pareto} shows.

Table~\ref{tab:boundary_main} computes the statistics of the economy of
Figure~\ref{f:fds_pareto} as the drift of $\ln A$ varies.
Table~\ref{tab:bounded} computes Example~\ref{eg:jadditive} with the same
productivity process, entry distribution and costs, with $\rho = 1/2$ and $k =
5 \times 10^{-5}$, and with sales in place of output. Income $y$ enters only
through the budget constraint, which fixes the mass of firms. None of the
reported ratios depends on $y$.

For both tables we compute the value function on a grid for $\ln \phi$ from
$-3$ to $20$ with step $0.0025$ and find $p^*$ by bisection on the entry
condition. We compute the occupation measure of \eqref{eq:pof} on the same
grid extended to $\ln \phi = 2{,}000$ for Table~\ref{tab:boundary_main} and to
$100$ for Table~\ref{tab:bounded}. In Table~\ref{tab:bounded} the index
thresholds $\phi = 10$ and $\phi = 100$ are about 3.6 and 36 times an
entrant's mean sales, within the range of firm sizes over which
\cite{chen2026growth} computes the same index from United States data.

The parameter $k$ sets the bound on sales. We set $k = 5 \times 10^{-5}$,
which puts the bound at about 1,700 times an entrant's mean sales. We measure
the sales index between 3.6 and 36 times an entrant's mean sales. Over that
range $pk\phi$ is at most about $0.01$ and the index exceeds $\alpha$ by less
than one percent, so Table~\ref{tab:bounded} reports an index close to
$\alpha$ in every row. For larger firms the index rises. At $\alpha = 1$ it is
about 1.03 between 36 and 100 times an entrant's mean sales and about 1.09
between 100 and 300 times. \cite{kondo2023heavy} find the same pattern for
employment, with a fitted index close to one up to about 1,000 employees and
rising above that size.\footnote{With $k$ a tenth as large, mean size and the
large-firm share at $\alpha = 1$ become 10.3 and 77 percent and the entry flow
0.53.}

\bibliographystyle{ecta}

\bibliography{bib_hht}

\begin{thebibliography}{50}
\newcommand{\enquote}[1]{``#1''}
\expandafter\ifx\csname natexlab\endcsname\relax\def\natexlab#1{#1}\fi

\bibitem[\protect\citeauthoryear{Acemoglu and Cao}{Acemoglu and
  Cao}{2015}]{acemoglu2015innovation}
\textsc{Acemoglu, D. and D.~Cao} (2015): \enquote{Innovation by entrants and
  incumbents,} \emph{Journal of Economic Theory}, 157, 255--294.

\bibitem[\protect\citeauthoryear{Alvarez and Stokey}{Alvarez and
  Stokey}{1998}]{alvarez1998dynamic}
\textsc{Alvarez, F. and N.~L. Stokey} (1998): \enquote{Dynamic programming with
  homogeneous functions,} \emph{Journal of Economic Theory}, 82, 167--189.

\bibitem[\protect\citeauthoryear{Arkolakis}{Arkolakis}{2016}]{arkolakis2016unified}
\textsc{Arkolakis, C.} (2016): \enquote{A Unified Theory of Firm Selection and
  Growth,} \emph{The Quarterly Journal of Economics}, 131, 89--155.

\bibitem[\protect\citeauthoryear{Axtell}{Axtell}{2001}]{axtell2001zipf}
\textsc{Axtell, R.~L.} (2001): \enquote{Zipf distribution of US firm sizes,}
  \emph{Science}, 293, 1818--1820.

\bibitem[\protect\citeauthoryear{Behrens and Murata}{Behrens and
  Murata}{2007}]{behrens2007general}
\textsc{Behrens, K. and Y.~Murata} (2007): \enquote{General Equilibrium Models
  of Monopolistic Competition: A New Approach,} \emph{Journal of Economic
  Theory}, 136, 776--787.

\bibitem[\protect\citeauthoryear{Benhabib and Bisin}{Benhabib and
  Bisin}{2018}]{benhabib2018skewed}
\textsc{Benhabib, J. and A.~Bisin} (2018): \enquote{Skewed wealth
  distributions: Theory and empirics,} \emph{Journal of Economic Literature},
  56, 1261--91.

\bibitem[\protect\citeauthoryear{Berlingieri, Jaimovich, Terry, and
  Vincent}{Berlingieri et~al.}{2025}]{berlingieri2025empirical}
\textsc{Berlingieri, G., N.~Jaimovich, S.~J. Terry, and N.~Vincent} (2025):
  \enquote{The empirical distribution of firm dynamics and its macro
  implications,} Working Paper 31337, National Bureau of Economic Research.

\bibitem[\protect\citeauthoryear{Boyd}{Boyd}{1990}]{boyd1990recursive}
\textsc{Boyd, J.~H.} (1990): \enquote{Recursive utility and the Ramsey
  problem,} \emph{Journal of Economic Theory}, 50, 326--345.

\bibitem[\protect\citeauthoryear{Bulow and Pfleiderer}{Bulow and
  Pfleiderer}{1983}]{bulow1983note}
\textsc{Bulow, J.~I. and P.~Pfleiderer} (1983): \enquote{A Note on the Effect
  of Cost Changes on Prices,} \emph{Journal of Political Economy}, 91,
  182--185.

\bibitem[\protect\citeauthoryear{Cao, Hyatt, Mukoyama, and Sager}{Cao
  et~al.}{2022}]{cao2022firm}
\textsc{Cao, D., H.~R. Hyatt, T.~Mukoyama, and E.~Sager} (2022): \enquote{Firm
  Growth through New Establishments,} Working paper, Georgetown University.

\bibitem[\protect\citeauthoryear{Carvalho and Grassi}{Carvalho and
  Grassi}{2019}]{carvalho2019large}
\textsc{Carvalho, V.~M. and B.~Grassi} (2019): \enquote{Large firm dynamics and
  the business cycle,} \emph{American Economic Review}, 109, 1375--1425.

\bibitem[\protect\citeauthoryear{Champernowne}{Champernowne}{1953}]{champernowne1953model}
\textsc{Champernowne, D.~G.} (1953): \enquote{A model of income distribution,}
  \emph{The Economic Journal}, 63, 318--351.

\bibitem[\protect\citeauthoryear{Chaney}{Chaney}{2008}]{chaney2008distorted}
\textsc{Chaney, T.} (2008): \enquote{Distorted Gravity: The Intensive and
  Extensive Margins of International Trade,} \emph{American Economic Review},
  98, 1707--1721.

\bibitem[\protect\citeauthoryear{Chen}{Chen}{2026}]{chen2026growth}
\textsc{Chen, Z.} (2026): \enquote{Economic growth and the rise of large
  firms,} \emph{Econometrica}, 94, 1375--1408.

\bibitem[\protect\citeauthoryear{Cirillo and H{\"u}sler}{Cirillo and
  H{\"u}sler}{2009}]{cirillo2009upper}
\textsc{Cirillo, P. and J.~H{\"u}sler} (2009): \enquote{On the upper tail of
  Italian firms’ size distribution,} \emph{Physica A: Statistical Mechanics
  and its applications}, 388, 1546--1554.

\bibitem[\protect\citeauthoryear{Clementi and Hopenhayn}{Clementi and
  Hopenhayn}{2006}]{clementi2006theory}
\textsc{Clementi, G.~L. and H.~A. Hopenhayn} (2006): \enquote{A theory of
  financing constraints and firm dynamics,} \emph{The Quarterly Journal of
  Economics}, 121, 229--265.

\bibitem[\protect\citeauthoryear{Dhingra and Morrow}{Dhingra and
  Morrow}{2019}]{dhingra2019monopolistic}
\textsc{Dhingra, S. and J.~Morrow} (2019): \enquote{Monopolistic Competition
  and Optimum Product Diversity under Firm Heterogeneity,} \emph{Journal of
  Political Economy}, 127, 196--232.

\bibitem[\protect\citeauthoryear{di~Giovanni, Levchenko, and
  Ranci{\`e}re}{di~Giovanni et~al.}{2011}]{digiovanni2011power}
\textsc{di~Giovanni, J., A.~A. Levchenko, and R.~Ranci{\`e}re} (2011):
  \enquote{Power Laws in Firm Size and Openness to Trade: Measurement and
  Implications,} \emph{Journal of International Economics}, 85, 42--52.

\bibitem[\protect\citeauthoryear{Ericson and Pakes}{Ericson and
  Pakes}{1995}]{ericson1995markov}
\textsc{Ericson, R. and A.~Pakes} (1995): \enquote{Markov-perfect industry
  dynamics: A framework for empirical work,} \emph{The Review of economic
  studies}, 62, 53--82.

\bibitem[\protect\citeauthoryear{Gabaix}{Gabaix}{2008}]{gabaix2008power}
\textsc{Gabaix, X.} (2008): \enquote{Power laws,} \emph{The New Palgrave
  Dictionary of Economics: Volume 1--8}, 5082--5085.

\bibitem[\protect\citeauthoryear{Gabaix}{Gabaix}{2011}]{gabaix2011granular}
---\hspace{-.1pt}---\hspace{-.1pt}--- (2011): \enquote{The granular origins of
  aggregate fluctuations,} \emph{Econometrica}, 79, 733--772.

\bibitem[\protect\citeauthoryear{Gabaix}{Gabaix}{2016}]{gabaix2016power}
---\hspace{-.1pt}---\hspace{-.1pt}--- (2016): \enquote{Power laws in economics:
  An introduction,} \emph{Journal of Economic Perspectives}, 30, 185--206.

\bibitem[\protect\citeauthoryear{Gaffeo, Gallegati, and Palestrini}{Gaffeo
  et~al.}{2003}]{gaffeo2003size}
\textsc{Gaffeo, E., M.~Gallegati, and A.~Palestrini} (2003): \enquote{On the
  size distribution of firms: additional evidence from the G7 countries,}
  \emph{Physica A: Statistical Mechanics and its Applications}, 324, 117--123.

\bibitem[\protect\citeauthoryear{Goldie}{Goldie}{1991}]{goldie1991implicit}
\textsc{Goldie, C.~M.} (1991): \enquote{Implicit renewal theory and tails of
  solutions of random equations,} \emph{The Annals of Applied Probability}, 1,
  126--166.

\bibitem[\protect\citeauthoryear{Gouin-Bonenfant and Toda}{Gouin-Bonenfant and
  Toda}{2019}]{gouin2019pareto}
\textsc{Gouin-Bonenfant, E. and A.~A. Toda} (2019): \enquote{Pareto
  Extrapolation,} Tech. rep., SSRN.

\bibitem[\protect\citeauthoryear{Hairer}{Hairer}{2018}]{hairer2006ergodic}
\textsc{Hairer, M.} (2018): \enquote{Ergodic properties of markov processes,}
  Tech. rep., University of Warwick.

\bibitem[\protect\citeauthoryear{Halter, Oechslin, and Zweim{\"u}ller}{Halter
  et~al.}{2014}]{halter2014inequality}
\textsc{Halter, D., M.~Oechslin, and J.~Zweim{\"u}ller} (2014):
  \enquote{Inequality and growth: the neglected time dimension,} \emph{Journal
  of Economic Growth}, 19, 81--104.

\bibitem[\protect\citeauthoryear{Hern{\'a}ndez-Lerma and
  Lasserre}{Hern{\'a}ndez-Lerma and Lasserre}{2012}]{hernandez2012further}
\textsc{Hern{\'a}ndez-Lerma, O. and J.~B. Lasserre} (2012): \emph{Further
  topics on discrete-time Markov control processes}, vol.~42, Springer Science
  \& Business Media.

\bibitem[\protect\citeauthoryear{Hopenhayn, Neira, and Singhania}{Hopenhayn
  et~al.}{2022}]{hopenhayn2022population}
\textsc{Hopenhayn, H., J.~Neira, and R.~Singhania} (2022): \enquote{From
  population growth to firm demographics: Implications for concentration,
  entrepreneurship and the labor share,} \emph{Econometrica}, 90, 1879--1914.

\bibitem[\protect\citeauthoryear{Hopenhayn and Rogerson}{Hopenhayn and
  Rogerson}{1993}]{hopenhayn1993job}
\textsc{Hopenhayn, H. and R.~Rogerson} (1993): \enquote{Job turnover and policy
  evaluation: A general equilibrium analysis,} \emph{Journal of political
  Economy}, 101, 915--938.

\bibitem[\protect\citeauthoryear{Hopenhayn}{Hopenhayn}{1992}]{hopenhayn1992entry}
\textsc{Hopenhayn, H.~A.} (1992): \enquote{Entry, exit, and firm dynamics in
  long run equilibrium,} \emph{Econometrica}, 1127--1150.

\bibitem[\protect\citeauthoryear{Hubmer and Restrepo}{Hubmer and
  Restrepo}{2026}]{hubmer2026not}
\textsc{Hubmer, J. and P.~Restrepo} (2026): \enquote{Not a typical firm:
  Capital--labor substitution and firms' labor shares,} \emph{American Economic
  Journal: Macroeconomics}, 18, 34--71.

\bibitem[\protect\citeauthoryear{Kacperczyk, Nosal, and Stevens}{Kacperczyk
  et~al.}{2019}]{kacperczyk2018investor}
\textsc{Kacperczyk, M., J.~Nosal, and L.~Stevens} (2019): \enquote{Investor
  sophistication and capital income inequality,} \emph{Journal of Monetary
  Economics}, 107, 18--31.

\bibitem[\protect\citeauthoryear{Kang, Jiang, Cheong, and Yoon}{Kang
  et~al.}{2011}]{kang2011changes}
\textsc{Kang, S.~H., Z.~Jiang, C.~Cheong, and S.-M. Yoon} (2011):
  \enquote{Changes of firm size distribution: The case of Korea,} \emph{Physica
  A: Statistical Mechanics and its Applications}, 390, 319--327.

\bibitem[\protect\citeauthoryear{Kimball}{Kimball}{1995}]{kimball1995quantitative}
\textsc{Kimball, M.~S.} (1995): \enquote{The Quantitative Analytics of the
  Basic Neomonetarist Model,} \emph{Journal of Money, Credit and Banking}, 27,
  1241--1277.

\bibitem[\protect\citeauthoryear{Klenow and Willis}{Klenow and
  Willis}{2016}]{klenow2016real}
\textsc{Klenow, P.~J. and J.~L. Willis} (2016): \enquote{Real Rigidities and
  Nominal Price Changes,} \emph{Economica}, 83, 443--472.

\bibitem[\protect\citeauthoryear{Kondo, Lewis, and Stella}{Kondo
  et~al.}{2023}]{kondo2023heavy}
\textsc{Kondo, I.~O., L.~T. Lewis, and A.~Stella} (2023): \enquote{Heavy
  Tailed, but not {Z}ipf: Firm and Establishment Size in the {U.S.}}
  \emph{Journal of Applied Econometrics}, 38, 767--785.

\bibitem[\protect\citeauthoryear{Lorden}{Lorden}{1970}]{lorden1970excess}
\textsc{Lorden, G.} (1970): \enquote{On excess over the boundary,} \emph{The
  Annals of Mathematical Statistics}, 41, 520--527.

\bibitem[\protect\citeauthoryear{Luttmer}{Luttmer}{2007}]{luttmer2007selection}
\textsc{Luttmer, E.~G.} (2007): \enquote{Selection, growth, and the size
  distribution of firms,} \emph{The Quarterly Journal of Economics}, 122,
  1103--1144.

\bibitem[\protect\citeauthoryear{Luttmer}{Luttmer}{2011}]{luttmer2011mechanics}
---\hspace{-.1pt}---\hspace{-.1pt}--- (2011): \enquote{On the mechanics of firm
  growth,} \emph{The Review of Economic Studies}, 78, 1042--1068.

\bibitem[\protect\citeauthoryear{Luttmer}{Luttmer}{2012}]{luttmer2012technology}
---\hspace{-.1pt}---\hspace{-.1pt}--- (2012): \enquote{Technology diffusion and
  growth,} \emph{Journal of Economic Theory}, 147, 602--622.

\bibitem[\protect\citeauthoryear{Melitz and Ottaviano}{Melitz and
  Ottaviano}{2008}]{melitz2008market}
\textsc{Melitz, M.~J. and G.~I.~P. Ottaviano} (2008): \enquote{Market Size,
  Trade, and Productivity,} \emph{The Review of Economic Studies}, 75,
  295--316.

\bibitem[\protect\citeauthoryear{Meyn and Tweedie}{Meyn and
  Tweedie}{2009}]{meyn2012markov}
\textsc{Meyn, S.~P. and R.~L. Tweedie} (2009): \emph{{M}arkov Chains and
  Stochastic Stability}, Cambridge: Cambridge University Press, 2nd ed.

\bibitem[\protect\citeauthoryear{Mr{\'a}zov{\'a} and Neary}{Mr{\'a}zov{\'a} and
  Neary}{2017}]{mrazova2017not}
\textsc{Mr{\'a}zov{\'a}, M. and J.~P. Neary} (2017): \enquote{Not So Demanding:
  Demand Structure and Firm Behavior,} \emph{American Economic Review}, 107,
  3835--3874.

\bibitem[\protect\citeauthoryear{Nirei}{Nirei}{2006}]{nirei2006threshold}
\textsc{Nirei, M.} (2006): \enquote{Threshold behavior and aggregate
  fluctuation,} \emph{Journal of Economic Theory}, 127, 309--322.

\bibitem[\protect\citeauthoryear{Shanker and Stachurski}{Shanker and
  Stachurski}{2026}]{shanker2026pareto}
\textsc{Shanker, A. and J.~Stachurski} (2026): \enquote{Pareto Tails in Models
  of Firm Entry and Exit,} Working paper, companion to the present paper.

\bibitem[\protect\citeauthoryear{Simon}{Simon}{1955}]{simon1955class}
\textsc{Simon, H.~A.} (1955): \enquote{On a class of skew distribution
  functions,} \emph{Biometrika}, 42, 425--440.

\bibitem[\protect\citeauthoryear{Sterk, Sedl{\'a}{\v{c}}ek, and Pugsley}{Sterk
  et~al.}{2021}]{sterk2021nature}
\textsc{Sterk, V., P.~Sedl{\'a}{\v{c}}ek, and B.~Pugsley} (2021): \enquote{The
  nature of firm growth,} \emph{American Economic Review}, 111, 547--579.

\bibitem[\protect\citeauthoryear{Zhang, Chen, and Wang}{Zhang
  et~al.}{2009}]{zhang2009zipf}
\textsc{Zhang, J., Q.~Chen, and Y.~Wang} (2009): \enquote{Zipf distribution in
  top Chinese firms and an economic explanation,} \emph{Physica A: Statistical
  Mechanics and its Applications}, 388, 2020--2024.

\bibitem[\protect\citeauthoryear{Zhelobodko, Kokovin, Parenti, and
  Thisse}{Zhelobodko et~al.}{2012}]{zhelobodko2012monopolistic}
\textsc{Zhelobodko, E., S.~Kokovin, M.~Parenti, and J.-F. Thisse} (2012):
  \enquote{Monopolistic Competition: Beyond the Constant Elasticity of
  Substitution,} \emph{Econometrica}, 80, 2765--2784.

\end{thebibliography}

\end{document}